\documentclass[english]{llncs}
\usepackage[T1]{fontenc}
\usepackage[latin9]{inputenc}
\usepackage{xcolor}
\usepackage{pdfcolmk}
\usepackage{array}
\usepackage{longtable}
\usepackage{float}
\usepackage{pmboxdraw}
\usepackage{amsmath}
\usepackage{amssymb}
\usepackage{graphicx}
\usepackage{setspace}
\PassOptionsToPackage{normalem}{ulem}
\usepackage{ulem}
\makeatletter

\providecommand{\tabularnewline}{\\}
\providecolor{lyxadded}{rgb}{0,0,1}
\providecolor{lyxdeleted}{rgb}{1,0,0}

\usepackage{leftidx}

\usepackage{type1cm}
\renewcommand\normalsize{%
   \@setfontsize\normalsize{13pt}{14.5pt}%
   \abovedisplayskip 12\p@ \@plus3\p@ \@minus7\p@
   \abovedisplayshortskip \z@ \@plus3\p@
   \belowdisplayshortskip 6.5\p@ \@plus3.5\p@ \@minus3\p@
   \belowdisplayskip \abovedisplayskip
   \let\@listi\@listI}\normalsize 

\usepackage{enumitem}
\renewcommand{\labelenumi}{(\arabic{enumi})}

\usepackage{changepage}

\makeatother

\usepackage{babel}
\begin{document}

\title{AntPaP: Patrolling and Fair Partitioning of Graphs by A(ge)nts Leaving
Pheromone Traces}

\author{Gidi Elazar and Alfred M. Bruckstein}

\institute{Multi Agent Robotic Systems (MARS) Lab \\
Technion Autonomous Systems Program (TASP) \\
Center for Intelligent Systems (CIS)\\
Department of Computer Science\\
 Technion, Haifa 32000, Israel}
\maketitle
\begin{abstract}
A team of identical and oblivious ant-like agents -- a(ge)nts -- leaving
pheromone traces, are programmed to jointly patrol an area modeled
as a graph. They perform this task using simple local interactions,
while also achieving the important byproduct of partitioning the graph
into roughly equal-sized disjoint sub-graphs. Each a(ge)nt begins
to operate at an arbitrary initial location, and throughout its work
does not acquire any information on either the shape or size of the
graph, or the number or whereabouts of other a(ge)nts. Graph partitioning
occurs spontaneously, as each of the a(ge)nts patrols and expands
its own pheromone-marked sub-graph, or region. This graph partitioning
algorithm is inspired by molecules hitting the borders of air filled
elastic balloons: an a(ge)nt that hits a border edge from the interior
of its region more frequently than an external a(ge)nt hits the same
edge from an adjacent vertex in the neighboring region, may conquer
that adjacent vertex, expanding its region at the expense of the neighbor.
Since the rule of patrolling a region ensures that each vertex is
visited with a frequency inversely proportional to the size of the
region, in terms of vertex count, a smaller region will effectively
exert higher ``pressure'' at its borders, and conquer adjacent vertices
from a larger region, thereby increasing the smaller region and shrinking
the larger. The algorithm, therefore, tends to equalize the sizes
of the regions patrolled, resembling a set of perfectly elastic physical
balloons, confined to a closed volume and filled with an equal amount
of air. The pheromone based local interactions of agents eventually
cause the system to evolve into a partition that is close to balanced
rather quickly, and if the graph and the number of a(ge)nts remain
unchanged, it is guaranteed that the system settles into a stable
and balanced partition.
\end{abstract}

\section{Introduction}

Patrolling is continuously traveling through an environment in order
to supervise or guard it. Although mostly used to refer to humans
guarding an area, the term patrolling is also used to describe surveying
through a digital, virtual environment. Consider, for example, the
task of repeatedly reading web pages from the world-wide-web in order
to keep an updated database representing the links between pages,
possibly for the purpose of later retrieval of pages in an a accurate
and prompt manner. These problems exhibit similarities, in the sense
that they can be represented as traveling through the vertices of
a graph. But there are also differences: a physical area is usually
fixed in size, whereas the virtual area is, in general, prone to constant
change. The number of human guards is, generally, fixed for the particular
area being patrolled, while the number of software agents or ``bots''
performing a large scale patrolling task may be subject to change
as well. 

Partitioning a graph into similar sized components is an important
and difficult task in many areas of science and engineering. To name
few examples, we can mention the partitioning of a netlist of an electronic
VLSI design \cite{VLSI}, the need for clustering in data mining \cite{Clustering Survey},
and the design of systems that balance the load on computer resources
in a networked environment \cite{Load Balancing}. 

The general graph partition problem is loosely defined as dividing
a graph into disjoint, connected components, such that the components
are \textit{similar }to each other in some sense. Practical considerations
impose additional constrains. For example, an important problem, known
as the \textit{graph k-cut}, requires a partition where the sum of
the \textit{weights }of vertices belonging to each component\textit{
}is more or less equal, and additionally, the number and/or the sum
of weights of edges that \textit{connect disjoint components} is minimized
\cite{k-cut}. The \textit{k-cut} problem can model the distribution
of tasks between computers on a network, while minimizing communication
requirements between them. 

In this work we define a patrolling strategy that fairly divides the
work of patrolling the environment among several a(ge)nts by partitioning
it into regions of more or less the same size. We have no constraints
on edges connecting different components, but we impose strict restrictions
on our patrolling agents in search for a heuristic multi-agent graph
partitioning algorithm that may continuously run in the background
of a host application. We are interested in programming the same behavior
for each individual ant-like agent, which should be very simple in
terms of resources, hardware or communications. Furthermore the agents
have no id's, hence are part of a team of units that are anonymous
and indistinguishable to each other. Our a(ge)nts should have very
little knowledge about the system or environment they operate in,
have no awareness on the size or shape of the graph, no internal memory
to accumulate information, nor a sense of the number and locations
of other agents active in the system. These limitations mean that
such a multi agent process has inherent scalability; the environment
might be large, complex, and subject to changes, in terms of vertices,
edges and even the number of agents, and our simple agents should
still be able to patrol it, while also evolving towards, and ultimately
finding balanced partitions, if such partitions exist. To simplify
the discussion, we will think of a graph to be partitioned as a planar
area, and the task at hand will be to partition the area into regions
of more or less the same size. The area is modeled as a grid, where
each vertex is a unit area, thus a balanced partition should have
components of roughly the same number of vertices. In our scheme,
agents are each given the task to \textit{patrol} and define a region
of their own, and have the ability to expand their region via conquests.
Like ants, our agents leave pheromone marks on their paths. The marks
decay with time and are subsequently used as cues by all the agents
to make decisions about their patrolling route and about the possibility
to expand their region. By assumption, each agent operates \textit{locally},
thus it can sense levels of pheromones or leave pheromone marks on
the vertex it is located, on its edges and on adjacent vertices. 

While patrolling its region, an agent visits a vertex and reads the
intensity of pheromone marks that remain from previous visits. It
then uses the reading, and the known rate of pheromone decay, to calculate
the vertex's idle-time -- the time that passed since the previous
visit. Using the decaying pheromone mark we can chose a patrolling
rule according to which the agents visit the vertices of their region
in repetitive cycles, each vertex being marked with a pheromone once
on each cycle. The patrolling process hence ensures that the idle-time
measured by agents on visits to their region's vertices is the same,
effectively encoding their region's \textit{cover time }-- the time
that takes for an a(ge)nt to complete a full patrolling cycle -- and
therefore it can also be used to estimate the region's size: the shorter
the cover time, the smaller the region. 

We assume that each agent detects pheromones without being able to
distinguish between them, except for recognizing its own pheromone.
When an a(ge)nt hits a border edge -- an edge that connects its region
with one that is patrolled by another agent - it can use the neighbor's
idle-time (encoded in its pheromone marks) to calculate the size of
the neighboring region, and thereby decide whether to try to conquer
the vertex ``on the other side of the border''. This causes an effect
that mimics pressure equalization between gas-filled balloons: at
two vertices on opposing ends of a border edge, the agent that hits
the border more frequently is the agent with a shorter cover time
(patrolling the smaller region) hence it may attempt a conquest. 

We define that in a balanced partition, any pair of neighboring regions
have a size difference of at most one vertex. This means that for
a graph $\mathcal{G}$ and $k$ agents, our partitioning heuristics
ensures a worst case difference of $k-1$ vertices between the largest
and smallest of the regions, once a balanced partition is reached.
For example in a graph of 1 million vertices (e.g. database entries,
each representing a web page) and 10 agents (network bots patrolling
the pages), this difference is truly negligible. Additionally, the
length of the patrolling path is predetermined, and is proportional
to the size of the region being patrolled, therefore when a balanced
partition is reached, the algorithm guarantees that the idleness of
any of the vertices of the graph is bounded by a number of steps equal
to $2\left(\left\lfloor \dfrac{\left|\mathcal{G}\right|}{k}\right\rfloor +\left(k-1\right)\right)-1$,
about twice the size of the largest possible region (Note that $\left|\mathcal{G}\right|$denotes
the number of vertices in the graph $\mathcal{G}$).

In Figure \ref{fig:evolution intro} one can see a series of snapshots
depicting 8 patrolling agents working to partition a 50x50 grid. The
first snapshot shows an early phase of the joint patrolling algorithm,
where agents already captured some of the vertices around their initial
random locations, in the second, the area is almost covered and most
of the vertices of the graph are being patrolled, the third exhibits
a phase when all the area is covered but the regions are not balanced,
and finally, the last snapshot shows a balanced partition that the
system evolved into. 

\begin{figure}[h]
\includegraphics[scale=0.3]{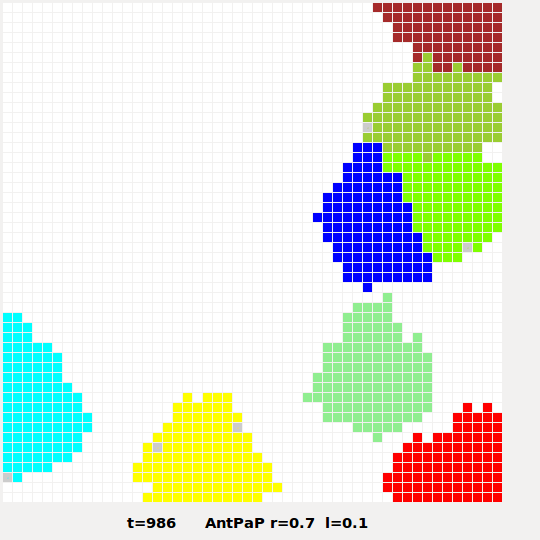}\includegraphics[scale=0.3]{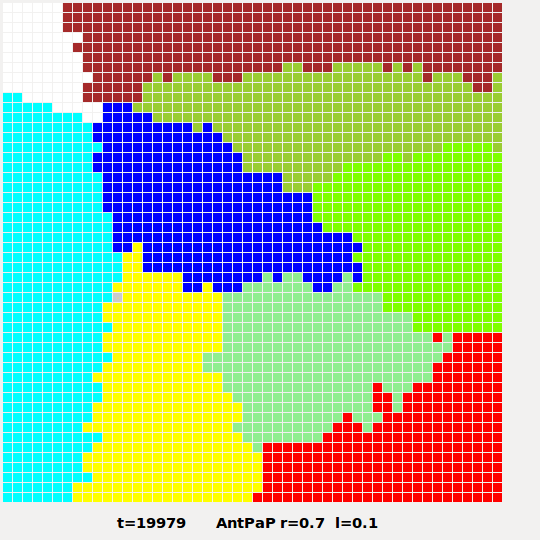}\bigskip{}
\includegraphics[scale=0.3]{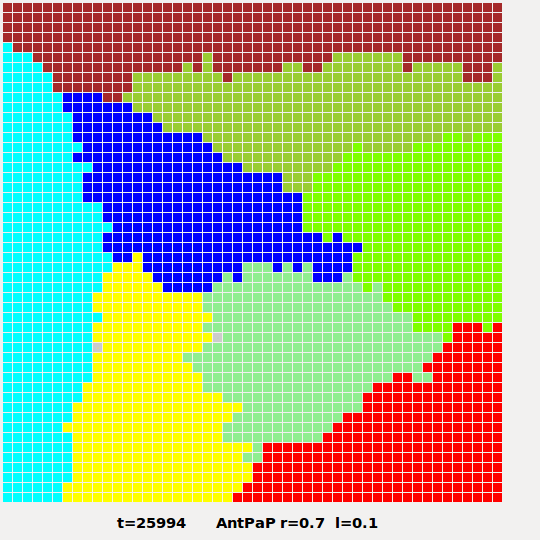}\includegraphics[scale=0.3]{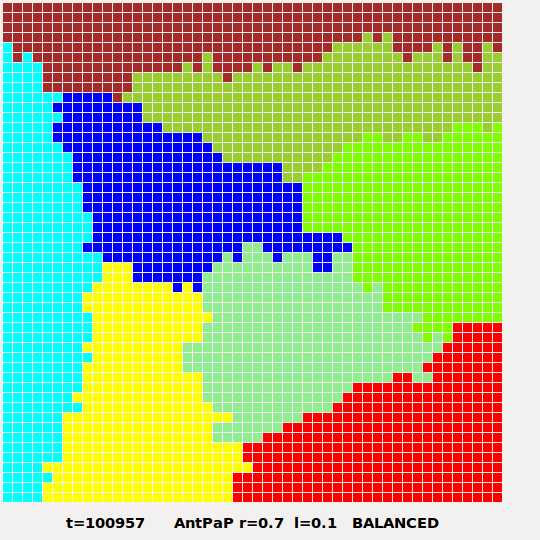}\caption{\label{fig:evolution intro} Evolution of a 50x50 Grid Graph Partitioning
by 8 A(ge)nts}
\end{figure}

This Figure exhibits typical stages in the evolution of such a system,
for which balanced partitions exist, and the environment graph remains
stationary for a time long enough for agents to find one of them.
Often, the agents will relatively quickly find a partition that covers
the graph, and is \textit{close} to being balanced. Then, on stationary
graphs, they may spend a rather long time to reach a perfectly balanced
partition. In a time varying environment the system will continuously
adapt to the changing conditions.

\section{Related Work}

The concept of partitioning a graph with a(ge)nts patrolling a region
and exerting pressure on neighboring regions was first presented by
Elor and Bruckstein in \cite{BDFS}. They proposed a patrol algorithm
named BDFS --Balloon DFS -- and this work is a follow up research
on this problem. According to BDFS, an agent patroling a smaller region
conquers vertices from a neighboring larger region. To achieve the
goal of patrolling an area, BDFS uses a variation of Multi-Level Depth-First-Search
(MLDFS), an algorithm presented by Wagner, Lindenbaum and Bruckstein
in \cite{MLDFS}. The task of the MLDFS too, was to distribute the
work of covering ``tiles on a floor'' among several identical agents.
The floor-plan mapping of the tiles is unknown and may even be changing,
an allegory for moving furniture around while agents are busy cleaning
the floor. MLDFS implements a generalization on DFS: agents leave
decaying pheromone marks on their paths as they advance in the graph,
and then use them either to move to the vertex least recently visited
or to backtrack. When none of the choices are possible, either when
the graph covering ends, or following to changes in the graph or loss
of tracks due to noise in the pheromone marks, agents \textit{reset,
}thus starting a new search. The time of reset, named ``the time
where new history begins'', is stored in the agents' memory, as a
\textbf{search-level} variable. After a \textit{reset, }the cycle
repeats, hence an agent traces pheromone marks left in an earlier
cycle. The mere existence of a pheromone mark is, however, not sufficient
for agents to choose a path not yet taken during the current search
cycle. To select the next step, agents use the value stored in the
\textbf{search-level} variable as a threshold: any pheromone that
was marked on a vertex or edge prior to this time must have been the
result of marking in an earlier cycle. In MLDFS, pheromones of all
agents are the same, and agents are allowed to step on each other's
paths. For the task of partition a graph, in BDFS each agent has its
own pheromone and it performs MLDFS cycles on its ``own'' region
of the graph, leaving its particular pheromone marks. As long as the
region is stationary, BDFS agents exactly repeat their previous route.
If the region changes, either expands or shrinks, it will cause BDFS
to look for a new and possibly substantially different route before
settling into the next search cycle. This occurs due to a subtlety
in the way that depth-first-search defines a spanning tree, a special
type of tree called a \textit{palm tree}, were each edge $\left(v,w\right)$
\uline{not} in the spanning tree connects a vertex $v$ with one
of its ancestors, see e.g. Tarjan \cite{DFS Tarjan Spanning Tree}.
The spanning tree defined during a BDFS search cycle does not consider
all edges emanating from \textit{all} of the region's vertices, simply
because some of the edges connect to vertices \textit{on neighboring
}regions\textit{.} When BDFS conquers a vertex, it is possible that
this vertex has more than one edge connecting to the region. All these
edges will now be considered during the next search cycle, a process
that may dictate a different palm tree. We call this event - \textit{respanning}.
In the algorithm we define here, named \textit{Ant Patrolling and
Partitioning, or AntPaP}, we use a different generalization of DFS
that avoids respanning. Furthermore we reduce the requirements on
the agent's capabilities. For example, our agents have no memory,
and also cannot control the levels of pheromones they leave, the pheromone
level at the time of a marking is always the same. We further add
the possibility for agents to \textit{lose} a vertex if a conquest
fails, and we provide a proof of convergence to balanced partitions,
while experimentally observing much faster evolution towards such
partitions. 

The subject of multi-agent patrolling has been extensively studied.
Lauri and Charpillet \cite{ACO applied to parolling} also use an
``ants paradigm'', where a method based on Ant Colony Optimization
(ACO), introduced by Dorigo, Maniezzo, and Colorni in \cite{ACO Ant System}.
ACO provides multi-agent solutions for various problems, for example
the Traveling Salesman Problem (TSP) in complete and weighted graphs
by a so-called \textit{ant-cycle} algorithm. Ants move to the next
vertex according to a probability that is biased by two parameters:
the closest neighbor vertex (corresponding to the lowest edge weight)
and the level of pheromone on the edge. During their search, ants
record their path to avoid visiting the same vertex twice. Since at
each step all ants traverse one edge to a neighboring vertex, all
ants complete their travel at the same time. Thereafter each ant leaves
pheromone marks on the entire path it took. Due to the probability
bias, shorter edges have a higher probability to be traversed, thus
it is probable that multiple ants traversed them, hence they tend
to accumulate stronger pheromone levels. The cycles repeat, and with
each cycle the biasing gets stronger towards the shortest path. The
process ends after a a-priori given number of cycles complete or when
the ants all agree on the shortest path. For the patrolling problem,
Lauri et. al. \cite{ACO applied to parolling} used this method to
find multiple paths, one for each agent, \textit{before} the agents
begin their joint work. Their algorithm employs multiple colonies
of ants where each agent is assigned one ant on each colony. Ants
in a colony cooperate (exchange information regarding their choices)
to divide the exploration into disjoint paths , leading the agents
to eventually cooperate in the patrolling task. Unlike for the TSP,
the environment graph is not required to be complete, and ants are
allowed to visit a vertex more than once when searching for a patrolling
route. 

Chevaleyre, Yann, Sempe, and Ramalho \cite{multi-agent patrolling strategies}
compared cyclic patrolling strategies, in which agents tend to follow
each other, to partitioning strategies, in which agents patrol each
its own region. By applying several algorithms on several graphs examples,
they found that the choice of strategy should be based on the shape
of the graph. The partitioning based strategy gets better results
on ``graphs having 'long corridors' connecting sub-graphs'', i.e.
if there are high weight edges that are slow to traverse, it is better
not traverse them at all by allocating them to connect disjoint partitions.

There is substantial research on heuristics for partitioning of a
graph, and some of it even related to multi-agent scenarios. Inspired
by ants, Comellas and Sapena \cite{aNTS} presented yet another \textit{Ants
}algorithm to find a \textit{k-cut} solution to a graph. The system
is initiated by randomly coloring all the graph vertices in a more
or less even number of colors and positioning the agents randomly
on the graph. Then a \textit{local cost} value is calculated for each
of the vertices, storing the percentage of neighbors that have the
same color as its own. Agents will then iteratively move to a neighboring
vertex $v$ that has the lowest cost (i.e. with the most neighbors
of a \textit{different }color than its own), and then switch colors
with a random vertex $u$ on the graph, where the color of $u$ is
the one most suitable for $v$, i.e. similar to most of $v$'s neighbors.
$u$ is selected from from a random list of vertices colored with
the same color as $u$, by choosing from this list the one with lowest
cost (most neighbors colored \textit{differently} than $u$). Then
the cost value is refreshed for both $v$ and $u$. On each iteration
the number of cuts, defined as the number of edges connecting vertices
of different colors, is calculated over all the edges of the graph
and the lowest value is stored. The choices of agent moves are stochastic,
i.e. agents have a probability $p$ to select the next vertex to move
to by using the cost value, otherwise it selects another neighbor
vertex at random. This allows the system to escape from local minima.
Unlike our algorithm, agents of Comellas and Sapena\textit{'s Ants
}aim to find a \textit{k-cut}, and while doing so do not leave pheromones
to be used as cues on vertices and edges they visit as our agents
do. Also, their\textit{ }agents are assumed to have the ability to
look at vertices that are anywhere in the graph and change their values,
thus their sensing is not local as in our algorithm. In \textit{Ants},
each iteration relies on a global calculation that involves access
to values on \textit{all} edges of the graph, in order to measure
the quality of the partition so far determined, as well as storing
the result. 

Inspired by bee foraging, McCaffrey \cite{bee colony} simulates a
bee colony in order to find a \textit{k-cut} graph partition. Each
of the agents, in this case called \textit{bees}, is assumed to know
in advance the size and shape of the graph, as well as the number
of components desired. The agent must have an internal memory to store
an array of vectors listing the vertices of all sub-graphs of a proposed
solution, as well as the number of cuts this partition has, as a measure
of its quality. In a hive, some 10\% of the bees are considered \textit{scouts,
}all other agents being in one of two states, \textit{active} or \textit{inactive.
}Emulated scouts select a random partition of the graph. If the selection
is better than what the scout previously found, it stores it in its
memory and communicates it to other bees in the hive that are in an
\textit{inactive} state. Some of those store the scout's solution
in their own memory, change their state to \textit{active} and begin
to search for a better partition around this solution. If an active
bee finds an improved solution it communicates it to the bees that
are left in the hive. After looking at neighboring solutions for a
long enough time, the \textit{active} bee returns to the hive and
becomes \textit{inactive} again. The algorithm, therefore, is constantly
searching for improvements in the quality of the partition that the
bees collectively determine. 

The partitioning and patrolling multi-agent algorithms that we have
surveyed above, all assume that agents posses substantial internal
memory. Some algorithms assume that the agents are able to sense and
even change values of vertices and/or edges in graph locations that
are distant from their position in the graph, and sometimes they can
even sense and/or store a representation of the whole graph in their
memory. Patrolling algorithms may be partition based, and then the
task is divided into two stages. In the first stage the graph is partitioned
into disjoint components, and at the second stage each of the agents
patrols one of those components. 

In our case, partitioning the graph, and thereby balancing the workload
among our agents, is a requirement. Our algorithm does not have stages,
the agents simply perform pheromone directed local steps thereby carrying
out a \textit{patrolling} algorithm, and while doing so also implicitly
cooperate in partitioning the graph. Our agents have no internal memory
at all. Their decisions are based on pheromone readings from vertices
and edges alone, and they can only sense or leave pheromone marks
around their graph location. One may view our solution for patrolling
and partitioning the graph environment as using the graph as a shared
distributed memory for our oblivious agents.

\pagebreak{}

\section{The AntPaP Algorithm and Empirical Results}

The task analyzed here is the partition of an area or environment
into regions of similar size by a set of agents with severe restrictions
on their capabilities. The inspiration for the algorithm are gas filled
balloons; consider a set of elastic balloons located inside a box,
and being inflated at a constant and equal rate, until the balloons
occupy the entire volume of the box. While inflating, it may be that
one balloon disturbs the expansion of another balloon. This may cause
a momentarily difference of the pressure in the balloons, until the
pressure difference is large enough to displace the disturbing balloon
and provide space for the expansion of the other. Since the amount
of gas is equal for all balloons, they will each occupy the same part
of the volume, effectively partitioning it into equal parts. Our agents
mimic this behavior by patrolling a region of the area ``of their
own'', while continuously aiming to expand it. The area is modeled
by a graph and the region is a connected component of the graph. When
expanding regions touch, the agent on the smaller region may conquer
vertices of the larger region. We assume that initially a given number
of agents are randomly placed in the environment, they start the process
of expanding and this process goes on forever. Eventually the expansion
is ``contained'' due to the interaction between the regions of the
agent, hence the process will lead to an equalization of the sizes
of the regions patrolled by the agents. In the discrete world of our
agents a partition to regions of exactly the same size may not exist,
therefore we define a balanced partition as such that any two neighboring
regions may have a size difference of at most one vertex. 

\medskip{}

\newpage{}

\noindent \textbf{\uline{Agent Modeling and Implementation Details\label{Agents model}}}

For simplicity, a(ge)nts operate in time slots, in a strongly asynchronous
mode, i.e. within a time interval every agent operates at some random
time, so that they do not interfere with each other. During a given
time slot, each agent may move over an edge to another vertex, and
may leave pheromone marks on a vertex and/or edges. The marks, if
made, are assumed to erase or coexist with the pheromone that remained
there from the previous visit. Agents have no control over the amount
of pheromone they leave, its initialization level being always the
same. Thereafter, the pheromone level decays in time. Each agent has
its own pheromone, thus pheromones are like colors identifying the
disjoint components and hence the partitioning of the graph. The agents
themselves can only tell if a pheromone is their own or not. Agents
are oblivious, i.e. have no internal memory. On each time slot, an
agent reads remaining pheromone levels previously marked on the vertex
it is located and its surroundings, and bases its decisions \uline{solely}
upon these readings. The readings and decisions are transient, in
the sense that they are forgotten when the time slot advances. Decaying
pheromone marks on vertices and edges linger, serving both as distributed
memory as well as means of communication. In our model, agents leave
pheromones in two patterns: one pheromone pattern is marked when agents
advance in their patrolling route, and the second pattern is used
when agents decide to remain on the same vertex. Pheromones are decaying
in time, thus once marked on a vertex or edge, their level on the
vertex or edge decreases with each time step. A straightforward way
for implementing such behavior in a computer program, is to use the
equivalent ``time markings'', i.e. stamping the \textit{time} at
which a pheromone is marked on the vertex or edge. We therefore denote
by $\varphi_{0}\left(v\right)$ the time of pheromone marking on vertex
$v$, hence $\varphi_{0}\left(v\right)=\tilde{t}$ means that an agent
left a pheromone on vertex $v$ at time $\tilde{t}$. As time advances,
the ''age'' of the pheromone on vertex $v$, i.e. the time interval
since it was marked, which can be calculated as $t-\varphi_{0}\left(v\right)$
where $t$ is the current time, advances as well. This is equivalent
to measuring the level of the temporally decaying pheromone on vertex
$v$, and using its value along with the known rate of decay to calculate
its ``age''. Similarly, $\varphi\left(u,v\right)=t$ is a time marking,
equivalent to the decaying pheromone level on the edge, where $\varphi\left(u,v\right)$
and $\varphi\left(v,u\right)$ are not necessarily the same. The use
of time markings require the computer program implementation to know
the current time $t$ in order to be able to calculate the age of
pheromones. However, the knowledge of current time is strictly limited
to its use in the emulation of temporally decaying pheromones by equivalent
time markings, thus it does not depart from our paradigm of obliviousness
and local decisions based on decaying pheromone markings only. When
an agent decides to leave a pheromone mark on a vertex, it may avoid
erasing the pheromone that remains from the previous (most recent)
visit. We denote the previous time marking as $\varphi_{1}\left(v\right)$,
thus when an agent marks a pheromone on vertex $v$, the computer
program implementation moves the value stored in $\varphi_{0}\left(v\right)$
to $\varphi_{1}\left(v\right)$ and afterwards sets the new time mark
to $\varphi_{0}\left(v\right)$. Hence, the value, $\varphi_{0}\left(v\right)-\varphi_{1}\left(v\right)$
encodes the \textit{idle time} of the vertex $v$. 

\noindent \textbf{\uline{The Patrol Algorithm}}

Agents patrol their region in a DFS-like route, in the sense that
they advance into each vertex once and backtrack through the same
edge once during a complete traversal of their region. When an agent
completes traversing the region it resets (i.e. it stays at the same
location for one time step and refreshes its pheromone mark), and
subsequently starts the search again. The cycles repeat the same route
as long as the region is unchanged. When the region does change --
either expanding or shrinking -- out agents persist on keeping advancing
and backtracking into a vertex through the same edge that was used
to conquer the vertex. This is implemented by marking \textit{``pair
trails'',} i.e. leaving pheromones over edges as well as vertices,
when conquering and (subsequently) advancing into a vertex. A pair
trail is a directed mark from a vertex $u$ to an adjacent vertex
$v$, of the form $\varphi_{0}\left(v\right)=\varphi\left(u,v\right)=\varphi\left(v,u\right)$,
and is one of the two pheromone patterns that agents leave. This behavior
results in a patrolling process that \textit{follows the pair trails},
were agents advance through the earliest marked pair trail, refreshing
the marks while doing so. When all pair trails to advance through
are exhausted, it backtracks through the same pair trail it entered.
An example of a route and the spanning tree it defines are depicted
in Figure \ref{fig:patrolling route and tree}. The departure from
the classic DFS is that edges that are not marked as pair trails are
ignored. The pair trails mark a \textit{spanning tree} (which is not
necessarily a palm tree) of the region, where its root is the vertex
where the search cycle begins, and each pair trail marks the path
advancing upwards the tree. When an agent backtracks to the root,
it has no untraveled pair trail to advance through, and it restarts
the search cycle remaining one time slot in the root. It then uses
the second marking pattern which is simply leaving a pheromone on
the vertex, denoted as $\varphi_{0}\left(u\right)=t$, where $u$
is the root.

\begin{figure}[h]
\makebox[1\columnwidth]{%
\includegraphics[scale=0.5]{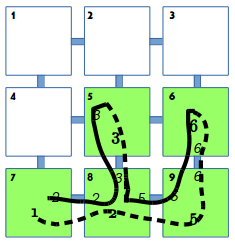}\hspace{4em}\includegraphics[scale=0.5]{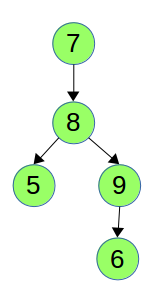}%
}

\caption{\label{fig:patrolling route and tree}An example of a patrolling route
in a region, and the spanning tree it defines. The arrows indicate
\textit{pair trail }directions: advance in the direction of the arrow,
and (eventually) backtrack the other way}
\end{figure}

Since agents advance and backtrack once from each vertex in their
region (except the root) and then restart a patrolling cycle in the
root, the number of steps in one patrolling cycle, called the \textit{cover
time}, is $\Delta t_{\alpha}=2\left|\mathcal{G}_{\alpha}\right|-1$.
$\mathcal{G}_{\alpha}$ denotes the region of agent $\alpha$, the
set of vertices that are part of $\alpha's$ patrolling cycle. Patrolling
cycles repeat the exact same route as long as the region remains unchanged,
hence the \textit{idle time} of any vertex $v$ of the region is also
the region's cover time $\varphi_{0}\left(v\right)-\varphi_{1}\left(v\right)=\Delta t_{\alpha}$.
Thus the pheromone markings on the vertex can be used to calculate
the size of the region $\left|\mathcal{G}_{\alpha}\right|=\dfrac{\Delta t_{\alpha}+1}{2}$.

\noindent \textbf{\uline{Conquest}}

To expand their region, agents may conquer vertices adjacent to (vertices
of) their region. For agent $\alpha$ to attempt to launch a conquest
from a vertex $u\in\mathcal{G}_{\alpha}$ to a target vertex $v$,
the following conditions must apply:
\begin{enumerate}
\item $v$ is not part of $\alpha$'s region, $v\notin\mathcal{G}_{\alpha}$,
let us then assume that $v\in\mathcal{G}_{\beta}$ of another agent
$\beta$ 
\item $u$ is subject to a \textit{double visit} by $\alpha$, i.e. $\alpha$
visits $u$ leaving pheromone marks \textit{twice}, while $v$ was
not visited even once by $\beta$ during the same period of time.
Since the time difference between the two visits by $\alpha$ is the
\textit{cover time} and the cover time is proportional to the size
of the region, it means that $\alpha$'s region is smaller than $\beta$'s,
$\left|\mathcal{G}_{\alpha}\right|<\left|\mathcal{G}_{\beta}\right|$.
An agent may check for this condition by evaluating if $\varphi_{1}\left(u\right)>\varphi_{0}\left(v\right)$. 
\item If the double visit condition is met, thus $\beta$'s region is larger,
allow a conquest attempt if it is not larger by \textit{exactly one}
vertex -- since a difference of one vertex is considered balanced. 
\item If the double visit condition is met and $\beta$'s region is larger
by \textit{exactly one} vertex, allow a conquest attempt if vertex
$v$ is stagnated -- it's pheromones are older than their purported
cover time. An agent checks this by comparing the idle time $t-\varphi_{0}\left(v\right)$
to the cover time $\varphi_{0}\left(v\right)-\varphi_{1}\left(v\right)$.
\end{enumerate}
Depending on the above conditions, an agent may stochastically \textit{attempt}
the conquest of vertex $v$, with a predefined probability $\mbox{0<\ensuremath{\rho_{c}}<1}$.
This mechanism works even if $v$ is not part of any of the other
agent's regions, $v\notin\mathcal{G}_{i},\ \forall i$. In such case
the pheromone marks on $v$ will never be refreshed and the conquest
conditions hold.

\noindent \textbf{\uline{\label{Temporary-Inconsistency}Temporary
Inconsistency}}

When a region expands or shrinks as result of conquests, its becomes
\textit{inconsistent }in the sense that the size of the region changed,
but at least some of the pheromone marks on its vertices encoding
the \textit{cover time} $\left(\varphi_{0}\left(v\right)-\varphi_{1}\left(v\right)\right)$
do not reflect that immediately. To regain consistency on a vertex
$u$, the pheromone marks on $u$ must be \textit{refreshed}, hence
an agent must leave there a fresh pheromone, and that may occur only
when the agent advances into vertex $u$ through a pair trail. Therefore,
there is a delay in the propagation of the change, thus there will
be a temporary inconsistency between the actual size of the region
and the cover time encoded on region's vertices. That inconsistency
is certainly not desirable since it might result in a miscalculation
of conquest conditions. Consider an agent $\alpha$ with a larger
region than two of its neighbors $\beta$ and $\gamma$. Both neighbors
will be attempting to conquer vertices from $\alpha$. Since all agents'
awareness is local, $\beta$ and $\gamma$ have no means to know that
$\alpha$ is shrinking due to the work of the other as well, and as
$\beta$ and $\gamma$ repeatedly conquer vertices from $\alpha$,
the combined conquests may accumulate to ``eat up'' too much out
of $\alpha$'s region up to a point where the imbalance is reversed,
and the areas of both $\beta$ and $\gamma$ are now larger than $\alpha$'s.
Nonetheless, the inconsistency is temporary. It is convenient to analyze
this issue by considering the spanning tree of pair trails. When an
agent conquers a vertex and expands its region, it results in adding
a leaf to the spanned tree, and losing a vertex to another agent results
in the \textit{pruning} of the tree, the splitting of the tree into
two or more branches, while the losing agent remains on one of them.
In either case, the \textit{follow the pair trails} strategy ensures
that the new route remains well defined. It is therefore sufficient
for an agent to patrol its region twice, to ensure that the region
is consistent, as described in the following Lemma:
\begin{lemma}
\label{lem:A-region-is-consistent}A region is consistent if it has
remained unchanged for a period of time which is twice its cover time.\end{lemma}
\begin{proof}
On the first cycle, the agent leaves a fresh pheromone, $\varphi_{0}$,
on each vertex, while the previous most recently visit, $\varphi_{1}$,
may reflect an inconsistent state. The second cycle repeats the exact
same route as the first, since the region remains unchanged, and now
both the most recent visit as well as the one preceding it, indicated
by pheromone levels $\varphi_{0}$ and $\varphi_{1}$, are updated,
thus $\varphi_{0}-\varphi_{1}$ reflects the cover time on all vertices
of the route and the region becomes consistent.\qed
\end{proof}
\noindent \textbf{\uline{Losing a Vertex}}

When balloons are inflated in a box, to the observer it looks as a
smooth evolution where the balloons steadily grow and occupy more
of the volume until the box is filled. But unlike gas inside a balloon
that exerts pressure in all directions concurrently, our discrete
agents work in steps, where at each step they attend one vertex of
their region, while the other vertices may be subject to conquests
by other agents. Since regions are defined by patrolling routes, an
agent $\alpha$, by conquest of a single vertex from $\beta$, may
prune $\beta$'s region in a way that leaves $\beta$ to patrol a
much smaller region effectively rendering it smaller than $\alpha$'s.
Now the ``balance tilts'', as the region that was larger prior to
the conquest becomes the smaller. Pruning may cut a spanning tree
into two or more sections, but in many cases the sections of the tree
may still be \textit{connected} by edges that are \textit{not }marked
by\textit{ }a pair trail. In such case, $\beta$ has an opportunity
to mark a pair trail over such an edge and regain access to a branch
still marked with its own pheromones. Yet, sometimes the pruning divides
the region into two unconnected components. We call these \textit{balloon
explosions, }and when these occur it is more difficult for the agent
that lost part of its region to regain its loss. Therefore, when an
agent launches a conquest attempt it is not always clear if its success
will advance or set-back the evolution towards convergence. It is
then natural to add the following \textit{vertex loss} rule: should
an agent fail the conquest attempt, there is a predefined probability
$0<\rho_{l}<1$ for \textit{losing} the vertex from which the attempt
was launched. \textit{Losing} the vertex may indeed be a better evolution
step than succeeding in that conquest, resembling actions of withdrawal
from local minima used in \textit{simulated annealing}. In fact, this
property becomes instrumental in our convergence proof for the AntPaP
algorithm. In order to prevent an agent from ``cutting the branch
it is sitting on'', we limit vertex loss events to steps of the patrolling
process in which agents backtrack, and, symmetrically restrict conquests
to steps in which the agents advance. 

We next list the algorithm describing the work of each agent on the
graph environment.

\begin{adjustwidth*}{-2cm}{-2cm}

\begin{singlespace}
\begin{tabular}{>{\raggedright}p{9cm}>{\raggedright}p{7cm}}
\begin{singlespace}
\textbf{\uline{Rule AntPaP}}\end{singlespace}
 & \tabularnewline
\textit{\footnotesize{}Entry point of an agent $\alpha$ at time step
$t$. }{\footnotesize \par}

\textit{\footnotesize{}Upon entry, the agent is located on vertex
$u$.}{\footnotesize \par}

{\footnotesize{}\smallskip{}
} & \tabularnewline
\uline{Pseudo Code}

{\footnotesize{}\smallskip{}
} & \uline{Description}\tabularnewline
\textbf{\footnotesize{}for each }{\footnotesize{}$\exists v\in\mathcal{N}\left(u\right)|\ \left(\mathcal{S}elf\left(v\right)=false\right)\wedge\left(\varphi_{1}\left(u\right)>\varphi_{0}\left(v\right)\right)$}{\footnotesize \par}

{\footnotesize{}\textSFxi{}}{\footnotesize \par}

{\footnotesize{}\textSFxi{}}{\footnotesize \par}

{\footnotesize{}\textSFxi{}}\textbf{\footnotesize{} if}{\footnotesize{}
$\left(\varphi_{0}\left(u\right)-\varphi_{1}\left(u\right)\right)+2\ne\left(\varphi_{0}\left(v\right)-\varphi_{1}\left(v\right)\right)$}{\footnotesize \par}

{\footnotesize{}\textSFxi{}}\textbf{\footnotesize{} $\ $ or}{\footnotesize{}
$\left(t-\varphi_{1}\left(v\right)\right)>\left(\varphi_{0}\left(v\right)-\varphi_{1}\left(v\right)\right)$
}\textbf{\footnotesize{}then}{\footnotesize \par}

{\footnotesize{}\textSFxi{} \textSFxi{}}{\footnotesize \par}

{\footnotesize{}\textSFxi{} \textSFxi{} }\textbf{\footnotesize{}if
}{\footnotesize{}$\nexists v\in\mathcal{N}\left(u\right)|\varphi\left(u,v\right)=\varphi_{0}\left(u\right)+1$
}\textbf{\footnotesize{}then}{\footnotesize \par}

{\footnotesize{}\textSFxi{} \textSFxi{} \textSFxi{} }\textbf{\footnotesize{}if
$\left(\mathcal{A}gent\mathcal{P}resent\left(v\right)=false\right)\wedge\left(x<\rho_{c}\right)$
then}{\footnotesize \par}

{\footnotesize{}\textSFxi{} \textSFxi{} \textSFxi{} \textSFxi{}}{\footnotesize \par}

{\footnotesize{}\textSFxi{} \textSFxi{} \textSFxi{} \textSFxi{} $\varphi\left(u,v\right)=t$;
$\varphi\left(v,u\right)=t$; $\varphi_{0}\left(v\right)=t$ }{\footnotesize \par}

{\footnotesize{}\textSFxi{} \textSFxi{} \textSFxi{} \textSFxi{} $\varphi_{1}\left(v\right)=0$ }{\footnotesize \par}

{\footnotesize{}\textSFxi{} \textSFxi{} \textSFxi{} \textSFii{}}\textbf{\footnotesize{}
goto}{\footnotesize{} $v$ ; }\textbf{\footnotesize{}return}{\footnotesize{} }{\footnotesize \par}

{\footnotesize{}\textSFxi{} \textSFxi{} \textSFxi{} $lose=\left(y<\rho_{l}\right)$ }{\footnotesize \par}

{\footnotesize{}\textSFii{} \textSFii{} \textSFii{}}\textbf{\footnotesize{}
else }{\footnotesize{}$lose=\left(y<\rho_{l}\right)$ }{\footnotesize \par}

{\footnotesize{}\bigskip{}
}{\footnotesize \par}

\textbf{\footnotesize{}if}{\footnotesize{} $\text{ }\exists v=\underset{v\in\mathcal{N}\left(u\right)}{argmin}\left\{ \varphi_{0}\left(v\right),s.t.\varphi_{1}\left(u\right)>\varphi_{0}\left(v\right),\ \mathcal{S}elf\left(v\right)=true\right\} $}\textbf{\footnotesize{}
then}{\footnotesize \par}

{\footnotesize{}\textSFxi{} $\varphi\left(u,v\right)=t$ ; $\varphi\left(v,u\right)=t$
; $\varphi_{0}\left(v\right)=t$ ; $\varphi_{1}\left(v\right)=0$}{\footnotesize \par}

{\footnotesize{}\textSFii{}}\textbf{\footnotesize{} goto}{\footnotesize{}
$v$ ; }\textbf{\footnotesize{}return}{\footnotesize{} }{\footnotesize \par}

{\footnotesize{}\bigskip{}
}{\footnotesize \par}

\textbf{\footnotesize{}if}{\footnotesize{} $\text{ }\exists v=\underset{v\in\mathcal{N}\left(u\right)}{argmin}\left\{ \varphi\left(u,v\right),s.t.\ \varphi\left(u,v\right)\ \textmd{in a pair trail}\right\} $}\textbf{\footnotesize{}
then}{\footnotesize \par}

{\footnotesize{}\textSFxi{} }{\footnotesize \par}

{\footnotesize{}\textSFxi{} }\textbf{\footnotesize{}if }{\footnotesize{}$\left(\varphi\left(u,v\right)=\varphi\left(v,u\right)=\varphi_{0}\left(v\right)\right)\wedge\left(\varphi_{0}\left(u\right)>\varphi_{0}\left(v\right)\right)$
}\textbf{\footnotesize{}then}{\footnotesize \par}

{\footnotesize{}\textSFxi{} \textSFxi{} $\varphi_{1}\left(v\right)=\varphi_{0}\left(v\right)$}{\footnotesize \par}

{\footnotesize{}\textSFxi{} \textSFxi{} $\varphi\left(u,v\right)=t$
; $\varphi\left(v,u\right)=t$ ; $\varphi_{0}\left(v\right)=t$ }{\footnotesize \par}

{\footnotesize{}\textSFxi{} \textSFii{} }\textbf{\footnotesize{}goto}{\footnotesize{}
$v$; }\textbf{\footnotesize{}return}{\footnotesize{} }{\footnotesize \par}

{\footnotesize{}\textSFxi{} }{\footnotesize \par}

{\footnotesize{}\textSFxi{} }\textbf{\footnotesize{}if}{\footnotesize{}
$\left(\varphi\left(u,v\right)=\varphi\left(v,u\right)=\varphi\left(u\right)\right)$
}\textbf{\footnotesize{}then }{\footnotesize \par}

{\footnotesize{}\textSFxi{} \textSFxi{}}\textbf{\footnotesize{} if
}{\footnotesize{}$lose$ }\textbf{\footnotesize{}then}{\footnotesize{} }{\footnotesize \par}

{\footnotesize{}\textSFxi{} \textSFxi{}}\textbf{\footnotesize{} }{\footnotesize{}\textSFii{}
$\varphi_{0}\left(u\right),\varphi_{1}\left(u\right)=0,\forall w\in\mathcal{N}\left(u\right),\varphi\left(u,w\right),\varphi\left(w,u\right)=0$}{\footnotesize \par}

{\footnotesize{}\textSFii{} \textSFii{} }\textbf{\footnotesize{}goto}{\footnotesize{}
$v$; }\textbf{\footnotesize{}return}{\footnotesize{} }{\footnotesize \par}

{\footnotesize{}\bigskip{}
}{\footnotesize \par}

{\footnotesize{}\bigskip{}
}{\footnotesize \par}

{\footnotesize{}$\varphi_{1}\left(v\right)=\varphi_{0}\left(v\right);\varphi_{0}\left(u\right)=t$}\\
\textbf{\footnotesize{}return} & \textbf{\footnotesize{}\uline{Explore border}}{\footnotesize{} }{\footnotesize \par}

{\footnotesize{}for each neighbor $v$ of $u$, }{\footnotesize{}\uline{not}}{\footnotesize{}
marked by $\alpha$'s pheromone, and meeting the}\textit{\footnotesize{}
double hit.}{\footnotesize \par}

{\footnotesize{}if size difference$\ne1$, }\\
{\footnotesize{}or $v$ is stagnated (pruned or empty) then,}{\footnotesize \par}

{\footnotesize{}\smallskip{}
}{\footnotesize \par}

{\footnotesize{}\quad{}if not backtracked to $u$}{\footnotesize \par}

{\footnotesize{}\quad{}\quad{}if chance allows conquest (x random
$\in\left[0,1\right]$)}\\
\textbf{\footnotesize{}\uline{Conquer vertex $v$}}{\footnotesize{} }{\footnotesize \par}

{\footnotesize{}mark pair-trail into $v$}{\footnotesize \par}

{\footnotesize{}move to $v$ and exit this step}{\footnotesize \par}

{\footnotesize{}\bigskip{}
}{\footnotesize \par}

{\footnotesize{}\smallskip{}
if no conquest, set }\textit{\footnotesize{}lose}{\footnotesize{}
flag (y random $\in\left[0,1\right]$)}{\footnotesize \par}

{\footnotesize{}\bigskip{}
}{\footnotesize \par}

\textbf{\footnotesize{}\uline{Rejoin Isolated}}{\footnotesize \par}

{\footnotesize{}if $v$ has $\alpha's$ (self) mark but }\textit{\footnotesize{}double
visit}{\footnotesize{} is met, rejoin $v$. Move to $v$ and exit
this step}{\footnotesize \par}

{\footnotesize{}\bigskip{}
}{\footnotesize \par}

{\footnotesize{}\bigskip{}
}{\footnotesize \par}

{\footnotesize{}\bigskip{}
}{\footnotesize \par}

{\footnotesize{}\medskip{}
}{\footnotesize \par}

{\footnotesize{}select $v$ of the oldest pair-trail}\\
{\footnotesize{}\medskip{}
}{\footnotesize \par}

\textbf{\footnotesize{}\uline{Advance}}{\footnotesize \par}

{\footnotesize{}if pair-trail points into $v$, and mark on $v$ is
older}{\footnotesize \par}

{\footnotesize{}keep previous time-mark}{\footnotesize \par}

{\footnotesize{}refresh the pair-trail pointing to $v$}{\footnotesize \par}

{\footnotesize{}move to $v$ and exit this step}{\footnotesize \par}

{\footnotesize{}\medskip{}
}{\footnotesize \par}

\textbf{\footnotesize{}\uline{Backtrack}}{\footnotesize \par}

{\footnotesize{}if pair-trail points into $u$ }{\footnotesize \par}

{\footnotesize{}\quad{}}\textbf{\footnotesize{}\uline{Lose}}{\footnotesize \par}

{\footnotesize{}\quad{}if }\textit{\footnotesize{}lose}{\footnotesize{}
flag is set}{\footnotesize \par}

{\footnotesize{}\quad{}remove time marks of $u$ and its pair-trails}{\footnotesize \par}

{\footnotesize{}move to $v$ and exit this step}{\footnotesize \par}

{\footnotesize{}\medskip{}
}{\footnotesize \par}

\textbf{\footnotesize{}\uline{Reset}}{\footnotesize \par}

{\footnotesize{}refresh vertex $u$}\\
{\footnotesize{}exit this step}\tabularnewline
\end{tabular}
\end{singlespace}

\end{adjustwidth*}

\pagebreak{}

\section{Typical Evolution of Patrolling and Partitioning}

During a patrolling cycle, an agent attempts conquests over all the
border edges of all vertices of its region. Hence its region may expand
with additional vertices bordering its route. This causes the spanning
tree defined by its DFS-route to have an ever growing number of branches
as the patrolling cycles continue, resulting in a tree shape resembling
a ``snow flake'' .

Additionally, agents have a strong tendency to form ``rounded''
regions, if the environment and other agents' regions allow it. This
happens because vertices that are candidates for conquest and are
adjacent to more than one of the region's vertices, have a higher
probability to be conquered and incorporated into the patrolled region.

\begin{figure}[h]
\makebox[1\columnwidth]{%
\includegraphics[scale=0.25]{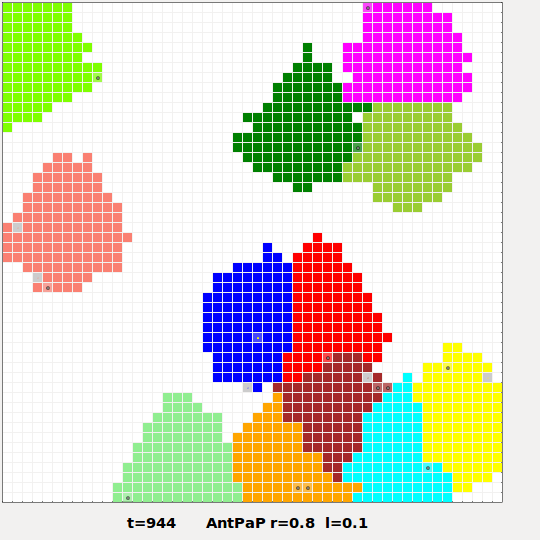}\includegraphics[scale=0.25]{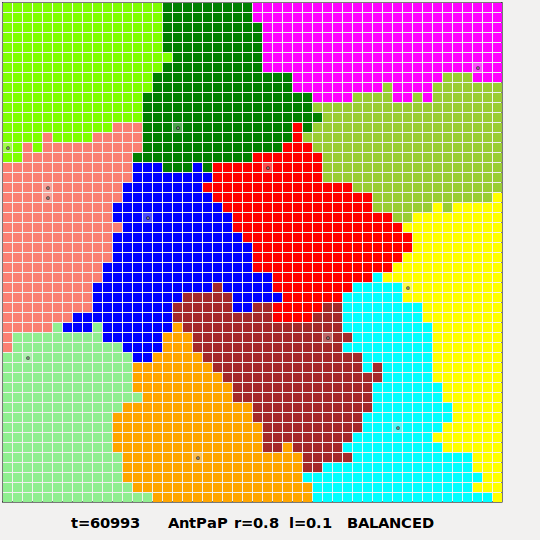}%
}

\caption{\label{fig:plausible rounded and thick}Rounded and Thick Regions
build-up by 12 ants on a 50x50 grid}
\end{figure}

These two properties, the snow-flake like spanning tree and rounded
build-ups, often assist in achieving a smooth evolution towards convergence.
Thick regions make the possibility of major ``balloon explosions''
unlikely. The thiner branches at the ends of the snow-flake-like region,
cause the pruning of the region by another agent to merely ``shave''
off small fragments from the region, and are less likely to cut out
a large portion of its vertices. Moreover, if a greater portion was
cut out by pruning, it is highly likely that the lost portion is connected
by edges that were not pair trails, making it is easier for the losing
agent to regain its vertices.

Figure \ref{fig:chart region size vs time} depicts a chart showing
a typically observed evolution of 12 regions in a 50x50 square grid.
The chart describes how the sizes of the regions change with time.
The images of Figure \ref{fig:plausible rounded and thick} are two
snapshots taken during the same evolution, the second snapshot being
of the balanced partition that the agents reached. The chart shows
that some regions grew faster than others, then at some point, the
regions grew enough so that the graph was covered (or close to being
covered due to continuous pruning), and the smaller regions began
to grow at the expense of the larger regions until they all reach
a very similar size, albeit not yet balanced. For practical purposes
arriving to this state in the graph might suffice, especially if the
graph is constantly changing and a partition that is balanced is not
well defined. This process occurs quickly, then there is a much longer
phase towards convergence. The ``turbulence'' seen at about $t=30,000$
are due to some mild ``balloon explosions'' followed by recovery
and convergence. The last \textit{plateau} towards convergence is
short in this particular example, but in some other simulated examples
it appeared much longer. This is especially noticeable in cases which
the evolution results in the system having two large regions that
are adjacent and have very similar sizes but are not yet balanced,
e.g. a size difference of 2 vertices. Such scenarios may increase
the time required for a double visit, which is a condition for conquest.

\begin{figure}[h]
\makebox[1\columnwidth]{%
\includegraphics[scale=0.5]{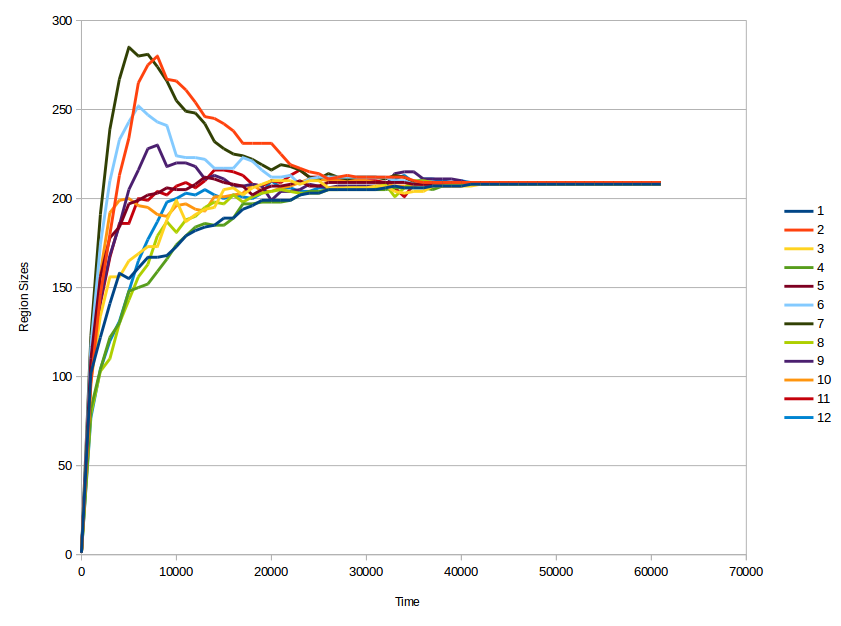}%
}

\caption{\textbf{\label{fig:chart region size vs time}}A Typically Observed
Partition Evolution , 12 agents on a 50x50 grid}
\end{figure}

Figure \ref{fig:convergence time as func of graph size} is a chart
depicting time to convergence for a system with 5 agents as a function
of the graph size, overlaying results of multiple simulation runs.
It shows that the spread in the time to convergence grows with size,
but it also clear that the majority, depicted as dense occurrences,
are not highly spread in value, indicative of runs of ``typical evolution
scenarios'' on square grid, as was the environment tested in these
simulations. The dotted line is an interpolation of average convergence
time among the results achieved for each graph size.

\begin{figure}[h]
\makebox[1\columnwidth]{%
\includegraphics[scale=0.6]{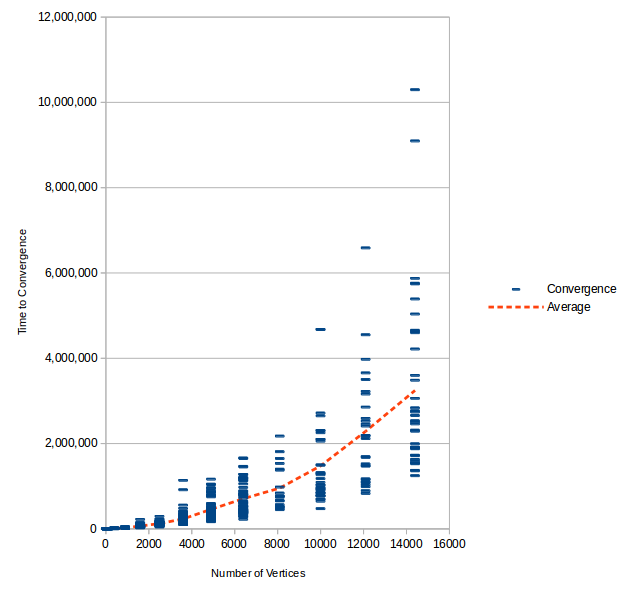}%
}

\caption{\label{fig:convergence time as func of graph size}Time to Convergence
as a Function of Graph Size}
\end{figure}

\section{Proof of Convergence}

The above discussed experimental evidence showcases an evolution of
the system towards a balanced partition (when the topology of the
environment/graph remains stationary), an evolution which is smooth
without ``dramatic'' incidents, driven by the AntPaP algorithm mimicking
pressure equalization. However, AntPaP is a heuristic process, and
the experimentally observed smooth convergence is by no means guaranteed.
In the evolution towards balanced partitions there are various events
that may substantially alter the size difference between regions,
and lead the system to longer and chaotic excursions. Chance dictates
the way regions expand and, for example, a region may build up with
less thickness in some areas allowing other agents to cut across it
causing ``major balloon explosions''. Furthermore, even a quite
well-rounded region may be subject to an ``unfair'' probabilistic
attack, driven to cut through its width and eventually succeeding
to remove a large portion of its area. To make things even worse,
the portion of the split region that ceased to be patrolled, becomes
a readily available prey to neighboring agents. Therefore, although
the system is relentlessly progressing towards a balanced partition
due to the rules of ``pressure equalization'', such ``balloon explosions''
are singular events that may significantly derail the smooth evolution
towards convergence, slowing the process considerably. One may wonder
if there are conditions were these events occur repeatedly, making
the convergence into a balanced partition an elusive target, that
may even never be reached. 

Clearly, there are systems where a balanced partition cannot be reached,
simply because one does not exist. An evident example is a graph having
the shape of a star. Consider a graph of 7 vertices, one at the center
and three branches of two vertices each. A system with 7 vertices
and 2 agents should be partitioned into two connected components,
one of 3 vertices and the other of 4. But such partition does not
exist, thus repeated balloon explosions will forever occur. Interestingly,
a balanced partition does exist for the same graph with 3 agents.

Therefore, our first step towards proof is to precisely define the
systems of interest, which are based on environment-graphs for which
a balanced partitions exist for \textit{any} number of agents (up
to a bound). The most general set of such graphs is an interesting
question in itself. For our purposes, we shall limit our analysis
to systems of the following type: 
\begin{enumerate}
\item The environment is a graph $\mathcal{G}$ that has a Hamiltonian path,
a path that passes through every vertex in a graph exactly once, with
any number of agents $n\le\left|\mathcal{G}\right|$ patrolling it.
Indeed, for any graph $\mathcal{G}$ that has a Hamiltonian path,
we can find multiple possibilities for a balanced partition for any
number of agents $n\le\left|\mathcal{G}\right|$. To name one, the
partition where each region includes vertices that are all adjacent
one to another along the Hamiltonian path, and the regions are chained
one after another along the path. Some of the regions can be of size
$\left\lfloor \dfrac{\left|\mathcal{G}\right|}{n}\right\rfloor $
and the others of size $\left\lfloor \dfrac{\left|\mathcal{G}\right|}{n}\right\rfloor +1$.
Note that for our purposes, the path does not need to be closed, so
the existence of a \textit{Hamiltonian cycle} is not required. 
\item The environment is a \textit{k-connected} graph, a graph that stays
connected when any fewer than $k$ of its vertices are removed, with
$n\le k$ agents patrolling it. In \cite{Gyori}, Gyori shows that
a \textit{k-connected} graph can always be partitioned into $k$ components,
including \textit{k} different and arbitrarily selected vertices.
\end{enumerate}
We shall analyze the evolution of the system as a stochastic process,
and base our proof of convergence on the theory of Markov chains.
The remainder of this section is organized as follows: we define a
``system configuration'' by considering a simple evolution example
and show that there always exists a mapping from configurations to
well defined ``states''. We then look at more complex configurations
and realize that although the set of configurations is not bounded,
it can be divided into a finite set of equivalence classes, each class
representing a state. Hence, we conclude that the number of states
in the Markov chain is finite, and the evolution of the configurations
maps into corresponding transitions between the states of the chain.
Next we use the concept of consistency of a region, as presented in
Lemma \ref{lem:A-region-is-consistent}, to conclude that if a balanced
partition is attained, it may persist indefinitely. This means that
balanced partitions map to recurrent states in the Markov chain. We
use this result to analyze the structure of the stochastic matrix
that describes the chain. Then we turn to prove that it is only the
balanced partitions that are mapped to recurrent states. We first
abstract the complexity of the problem by classifying all possible
graph partitions into mutually exclusive classes: uncovered, covered
but unbalanced, balanced but unstable, balanced and stable. Then we
proceed to analyze the changes that may cause the system to shift
from a configuration in one class to a configuration in another. Finally
we show that when the graph has a Hamiltonian path or is \textit{k-connected},
despite the possibility that the system may repeatedly transition
between these classes, it cannot do so indefinitely and will inevitably
have to sink into a recurrent state that belongs to a set of states
which are all assigned to the same balanced partition, forming a so-called
``recurrent class''.

\begin{definition}
The vertices and edges of the graph with their respective pheromone
markings, and the agent locations, will be called a configuration
of the system, and denoted by $\mathcal{C}$. Recall that, as discussed
in section \ref{Agents model} (sub-section ``Agents''), more straightforward
time markings, (in which the current time, $\mathcal{\varphi}_{0}\left(v\right)=t$,
is marked), can be used to emulate their equivalent pheromone markings
with temporal decay.
\end{definition}

\begin{figure}[h]
\makebox[1\columnwidth]{%
\includegraphics[scale=0.5]{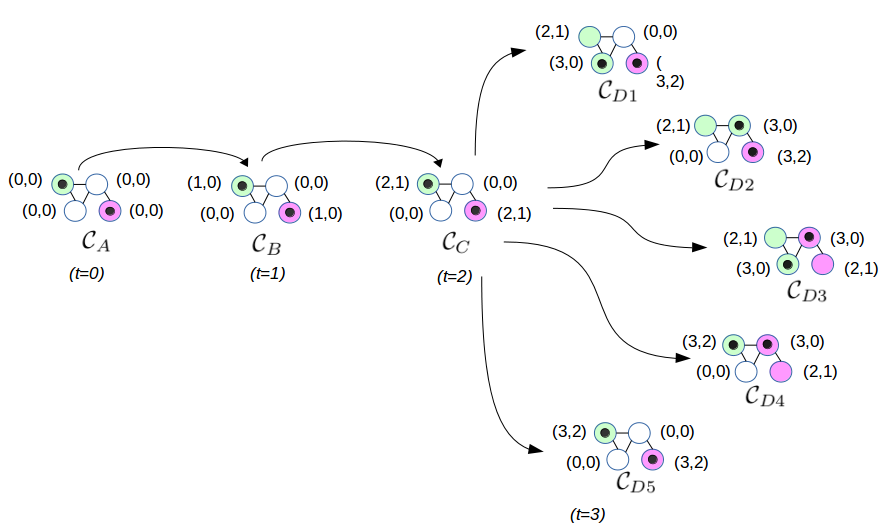}%
}\caption{\label{fig: 4 first steps} Starting at configuration\textbf{ }$\mathcal{C}_{A}$,
the system must move to $\mathcal{C}_{B}$ and then to $\mathcal{C}_{C}$.
Then it will stochastically transition to $\mathcal{C}_{D1}$, $\mathcal{C}_{D2}$,
$\mathcal{C}_{D3}$, $\mathcal{C}_{D4}$ or $\mathcal{C}_{D5}$.\protect \\
$\mathcal{C}_{C}$ and $\mathcal{C}_{D5}$ are equivalent since both
have the same pheromone decay levels. The values shown are time markings
$\left(\varphi_{0},\varphi_{1}\right)$.}
\end{figure}

The diagram of Figure \ref{fig: 4 first steps} depicts an example
of transitions between configurations of a system of a 4-vertex graph
on which two agents, the \textit{green} and the\textit{ cyan}, are
active\textit{. }$\mathcal{C}_{A}$ is the initial configuration at
\textit{$t=0$ }. Agents, shown as dots, are placed at some random
initial vertices that are colored according to the agent patrolling
them, \textit{green} at the top-left vertex and the\textit{ cyan}
at the bottom-right vertex. The pheromone markings on vertices are
shown as ordered pairs of time markings $\left(\varphi_{0},\varphi_{1}\right)$,
where $\varphi_{0}$ is the most recent time that a vertex was marked
with pheromone, and $\varphi_{1}$ is the previous time that the vertex
was marked (so, generally, $\varphi_{0}>\varphi_{1}$). When the two
agents \textit{wake up} at time slot $t=1,$ the readings of pheromone
marks around are all zero, therefore the \textit{double visit} condition
$\varphi_{1}\left(u\right)>\varphi_{0}\left(v\right)$, is not met,
and conquests are prohibited. Hence the only possible action for the
agents is \textit{reset}, i.e. leaving a fresh pheromone mark $\varphi_{0}\left(u\right)=1$
(the current time), and thus transitioning to the new configuration
$\mathcal{C}_{B}$. At $t=2$, the double visit condition is again
not met, and the system transitions to $\mathcal{C}_{C}$ (recall
that according to the AntPaP algorithm, at the time that a pheromone
is marked on a vertex, the previous mark is moved from $\varphi_{0}$
to $\varphi_{1}$). At $t=3$, conquest conditions are met for both
agents, and the system may now transition to any one of the configurations
$\left\{ \mathcal{C}_{D1},\mathcal{C}_{D2},\mathcal{C}_{D3},\mathcal{C}_{D4}\right\} $,
according to whether one or more conquests succeed, or to $\mathcal{C}_{D5}$,
with a probability of $\left(1-\rho_{c}\right)^{3}$, if all 3 conquest
attempts fail . It is important to notice that $\mathcal{C}_{D5}$
is equivalent to $\mathcal{C}_{C}$ (and in fact the configurations
are identical in terms of temporally decaying pheromone markings since
$\left(3,2\right)_{\textnormal{at }t=3}\equiv\left(2,1\right)_{\textnormal{at }t=2}\equiv\left(t,t-1\right)_{\textnormal{at }t}$).
In both we have a $\varphi_{0}$ pheromone that has been freshly marked,
so $\left(t-\varphi_{0}\right)=0$ , and a $\varphi_{1}$ which has
been marked on the immediately preceding time slot, hence $\left(t-\varphi_{1}\right)=1$.
We can, therefore, map each configuration of Figure \ref{fig: 4 first steps}
to distinct ``states'', states $1,2,\ldots,7$, as depicted in Figure
\ref{fig:4 first states}, and group $\mathcal{C}_{D5}$ and $\mathcal{C}_{C}$
into an \textit{equivalence class} of system configurations, were
configurations in such equivalence class map to the same ``state''
(in this case the equivalence class that includes $\mathcal{C}_{D5}$
and $\mathcal{C}_{C}$ maps to state 3).

\begin{figure}[h]
\makebox[1\columnwidth]{%
\includegraphics[scale=0.75]{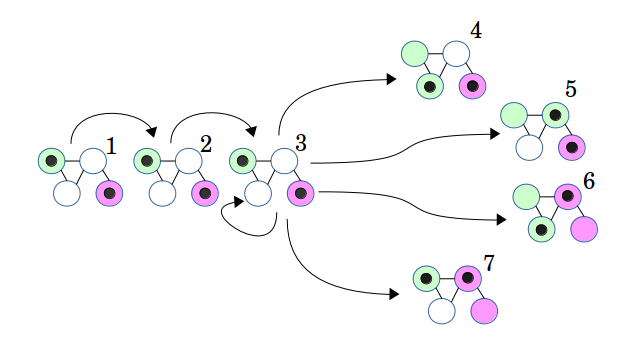}%
}\caption{\label{fig:4 first states} States Matching the Configurations of
Figure \ref{fig: 4 first steps} }
\end{figure}

Calculating transition probabilities between ``states'' is straightforward,
but sometimes subtleties arise, for example: 

$\mathcal{P}\left(State\ 1\ \to\ State\ 2\right)=$$\mathcal{P}\left(State\ 2\ \to\ State\ 3\right)=1$,

$\mathcal{P}\left(State\ 3\ \to\ State\ 3\right)=\left(1-\rho_{c}\right)^{3}$
(both agents attempt conquests, but fail), and

$\mathcal{P}\left(State\ 3\ \to\ State\ 7\right)=\frac{1}{2}\left(\rho_{c}\right)(1-\rho_{c})+\frac{1}{2}\left(1-\rho_{c}\right)^{2}\rho_{c}$
(due to the strong asynchronous assumption, the \textit{cyan} agent
has a probability of $\frac{1}{2}$ to move first within the time
slot, and if so, its attempted conquest has a probability $\rho_{c}$
to succeed. The result is multiplied by a probability $(1-\rho_{c})$
that the \textit{green} agent fails in its conquest. If the \textit{green
}moves first, there is a probability $\left(1-\rho_{c}\right)^{2}$
that it would fail both attempts, and then a probability $\rho_{c}$
that the \textit{cyan }succeeds).
\begin{definition}
An edge $\left(u,v\right)$ which is part of a pair-trail pattern
marking, so that $\varphi\left(u,v\right)=\varphi\left(v,u\right)=\varphi\left(v\right)$,
will be called a pair-trail edge.
\end{definition}

\begin{definition}
For a time interval in which no conquests or losses of a vertex in
a region occur, the region is considered ``stable''.\end{definition}
\begin{lemma}
\label{lem:consistent and zero =00003D M states}A system comprising
a graph $\mathcal{G}$ with time-invariant topology, and $k$ agents,
in a configuration $\mathcal{C}_{N}$, satisfying
\begin{enumerate}
\item The $k$ regions marked by the agents are consistent
\item Pheromone marks exist only on vertices and pair-trail edges inside
the k-regions, and no pheromone markings exist elsewhere on the graph,
\end{enumerate}
\noindent will transition through a finite sequence of M of states,
where M is the least common multiple of the cover times of the k-regions
(i.e. $M=lcm\left(\Delta t_{1},\Delta t_{2},\ldots,\Delta t_{k}\right))$,
prior arriving to a configuration $\mathcal{C}_{M}$ equivalent to
$\mathcal{C}_{N}$ (i.e. $\mathcal{C}_{M}\sim\mathcal{C}_{N})$, as
long as the k regions are stable (no conquests or losses occur). \end{lemma}
\begin{proof}
A consistent region $\mathcal{G}_{i}$ is a region for which $\varphi_{0}\left(v\right)-\varphi_{1}\left(v\right)=\Delta t_{i},\ \forall v\in\mathcal{G}_{i}$,
where $\Delta t_{i}$ is the cover time of region $\mathcal{G}_{i}$
(see Lemma \ref{lem:A-region-is-consistent}). Therefore, as long
as the region is stable, all vertices and all pair-trail edges of
$\mathcal{G}_{i}$, cyclically return to the exact same decaying pheromone
levels, i.e. exactly the same temporal differences $t-\varphi_{0}\left(v\right)$
(where $t$ is the current time) every $\Delta t_{i}$ steps. We can
thereby consider the $k$ regions in the partition, each repeating
its pheromone level markings independently, as $k$ cyclic processes
each with its own cycle time. Hence, all the processes complete an
integer number of cycles every $M$ steps, where $M$ is the least
common multiple of the cycle times, $M=lcm\left(\Delta t_{1},\Delta t_{2},\ldots,\Delta t_{k}\right)$,
which ensures that all vertices of all the $k$ regions in the configuration
\textit{$\mathcal{C}_{M}$} that was reached have exactly the same
temporal differences as in \textit{$\mathcal{C}_{N}$, }and therefore
\textit{$\mathcal{C}_{M}\sim\mathcal{C}_{N}$.} \qed
\end{proof}

In the above Lemma, we required to have no pheromones at all on edges
that are not pair-trail edges. But, if there were markings on such
edges, the patrolling agents would simply ignore them, according to
the AntPaP algorithm. Therefore such markings have no influence on
the possible future evolutions of the system. We shall formally define
states of the system by grouping together configurations that have
``the same future evolutions'', i.e. same possible future configuration
transition sequences with the same probabilities. For example, as
seen above, configurations that differ only by levels of pheromones
on non pair-trail edges form such equivalence classes, hence each
class defines a distinct state. In systems theory, this is the classical
\textit{Nerode} equivalence way of defining states.

Accordingly, two configurations that do not have the exact same patrolling
routes (either not having the same $k$ regions, or the agents have
developed different patrolling paths within the regions) cannot have
the same future evolutions, since, even without any conquests or loses,
the future sequences of configurations that the systems go through
are different due to the different patrolling steps. Therefore these
two configurations can not belong to the same equivalence class thus
represent distinct states.

Next we turn to discuss pheromone markings that may exist on vertices
and edges that are \uline{not} part of any current patrolling route,
hence outside of all the regions. Such scenarios may occur as result
of a successful conquest by an agent that disconnects the region of
a neighbor and hence prunes the spanning tree of that agent, splitting
it into two or more disjoint branches. Clearly, the latter agent remains
on one of these branches, while the others cease to be part of its
patrolling route and remain ``isolated''. 
\begin{definition}
A segment of a spanning tree (i.e. a set of vertices marked with pheromones
and connected by pair trails) that is \uline{not} part of a patrolling
route, thus not included in any of the regions, forms what we shall
call an isolated branch. For completeness, a single such vertex that
is not connected by a pair-trail is also considered an isolated branch.
\end{definition}

In our pheromone marking model we have not limited the pheromone decay,
thus, on an isolated branch, pheromones may decay indefinitely. This
means that there is no bound to the set of configurations, and raises
the question of whether there exists a bound to the set of equivalence
classes to which they can belong, hence a bound on the number of states
of the system. We shall, therefore, consider configurations that include
isolated branches, and analyze the effect of pheromone decay in these
branches on the evolution of the system, or more precisely, how such
decay influences ``future'' system states. We have already seen
that two configurations that do not have the exact same patrolling
routes must represent different states, thus we shall verify that
this distinction, by itself, does not produce an unbounded number
of states. The number of permutations of possible $k$ stable regions
is finite (in a finite graph), and for each such permutation, the
number of permutations of possible routes for the $k$ agents must
be finite too (since each of the $k$ regions have a finite number
of edges). 

We are therefore left to show that starting at any arbitrary configuration
with $k$ agents patrolling $k$ regions that also include isolated
branches, all future evolutions in an arbitrarily large interval in
which all $k$ regions remain stable, can be grouped into a finite
number of equivalence classes.

Let us consider a setup of $k$ regions and an isolated branch, where
the regions are stable in an arbitrarily large interval, and further
assume that a vertex $v$ of the isolated branch is adjacent to a
vertex $u$ in one of the regions (See figure \ref{fig:isolated branch prunning}).
Since the regions are stable, every patrolling cycle the value $\varphi_{0}\left(u\right)$
is refreshed, thus its time-marking increases with each cycle. On
the other hand, the time marking $\varphi_{0}\left(v\right)$ of the
vertex in the branch remains unchanged. When the agent is on vertex
$u$ the following scenarios may arise:

\begin{figure}[h]
\makebox[1\columnwidth]{%
\includegraphics[scale=0.55]{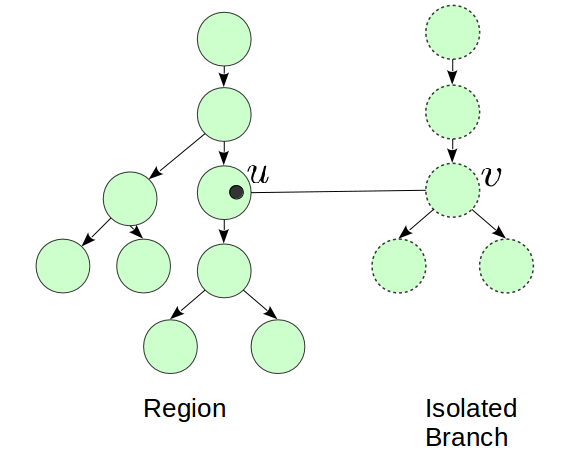}%
}

\caption{\label{fig:isolated branch prunning} An agent (black dot) on a vertex
$u$ adjacent to a vertex $v$ on an isolated branch that is marked
with the agent's pheromone. The solid line represents an edge, the
arrows represent pair-trail edges. Rejoining the lost vertex $v$
on the isolated branch results in pruning the isolated branch into
disjoint parts. One with the upper two vertices, another with the
lower three.}
\end{figure}

\begin{enumerate}
\item The vertex $v$ on the branch might be marked with the same agent's
pheromone, and hence moving into the adjacent vertex consists of the
action of \textit{rejoining }a vertex previously lost. According to
AntPaP, agents check for a double visit condition, i.e. $\varphi_{1}\left(u\right)>\varphi_{0}\left(v\right)$,
prior to this action. When traversing onto the vertex, e.g. at time
$t$, the agent marks there a fresh pheromone $\varphi_{0}\left(v\right)=t$.
This may result in splitting the isolated branch into two or more
disjoint branches. The agent will then follow the pair-trails emanating
from vertex $v$ at which it is presently located, oblivious to the
fact that pheromone marks on pair-trails and vertices of the branch
are old. Thereafter the agent traverses the section of the previously-isolated
branch it is located on, thus refreshing its marks, until all the
section is visited (e.g., in the example of Figure \ref{fig:isolated branch prunning},
the agent will visit all vertices on the lower section), then it returns
to the vertex $u$ from which the conquest was launched. The other
disjoint branches (in the example of Figure \ref{fig:isolated branch prunning},
the upper section) remain isolated.
\item if the branch is marked with another agent's pheromone, and conquest
conditions are met (e.g. double visit ($\varphi_{1}\left(u\right)>\varphi_{0}\left(v\right)$)
and the regions size difference is not exactly one vertex ( or equivalently
the difference in cover time is not 2 , i.e., $\left(\varphi_{0}\left(u\right)-\varphi_{1}\left(u\right)\right)+2\ne\left(\varphi_{0}\left(v\right)-\varphi_{1}\left(v\right)\right)$
), the agent may attempt a conquest on the vertex $v$ and thereafter
on all the vertices of the branch, one by one.
\item if the branch is marked with another agent's pheromone but the double
visit condition is \uline{not} met (i.e. $\varphi_{1}\left(u\right)\ngtr\varphi_{0}\left(v\right)$),
it may remain so only for at most two cycles of patrolling. Note that
$\varphi_{1}\left(u\right)$ is growing with each agent's visit, while
$\varphi_{0}\left(v\right)$ remains unchanged and as a result a double
visit condition will necessarily arise. Furthermore, meeting the double
visit condition also ensures that all the additional conquest conditions
are met at the same time, since either the cover time encoded in the
vertex of the isolated branch indicates a region size conducive to
conquests, or the agent recognizes that the neighboring region is
stagnated (not being patrolled for too long, i.e., $\left(t-\varphi_{1}\left(v\right)\right)>\left(\varphi_{0}\left(v\right)-\varphi_{1}\left(v\right)\right)$).
Therefore any further decay of the pheromone mark on the isolated
branch will not influence the future behavior of the system.
\end{enumerate}
A double visit is, therefore, a sufficient condition for an agent
to conquer or rejoin a vertex on an adjacent isolated branch, hence
we conclude the following:
\begin{lemma}
\label{lem:N+M states}A system comprising a graph $\mathcal{G}$
with time-invariant topology, and $k$ agents, in a configuration
$\mathcal{C}$ that includes exactly one isolated branch will transition
through a finite sequence of at most $N+M$ states, where $N=2*max\left(\Delta t_{1},\Delta t_{2},\ldots,\Delta t_{k}\right)$
and $M=lcm\left(\Delta t_{1},\Delta t_{2},\ldots,\Delta t_{k}\right)$,
where $\Delta t_{i}$ is the cover time of region $\mathcal{G}_{i}$,
as long as all the regions are stable (no conquests or losses occur). \end{lemma}
\begin{proof}
The completion of two patrolling cycles of a region by its patrolling
agent ensures that the double visit condition is met at any vertex
$u$ of the region adjacent to a vertex $v$ of the isolated branch
(see discussion above). Therefore, after an interval of \textit{$N=2*max\left(\Delta t_{1},\Delta t_{2},\ldots,\Delta t_{k}\right)$
}(i.e. when the agent on the largest region completed two patrolling
cycles) it is certain that the double visit condition is globally
met (i.e. for any vertex $u$ on any of the regions adjacent to any
vertex $v$ in the isolated branch). Moreover, it will be met on any
time step that follows (as long as the regions are stable). Hence
any two configurations on which the double visit condition is globally
met, and have the same levels of pheromones on vertices and pair-trail
edges that are in the $k$ regions (but may differ in levels of pheromones
on the isolated branch) are equivalent.

Since we also know, based on Lemma \ref{lem:consistent and zero =00003D M states},
that every\\
 $M=lcm\left(\Delta t_{1},\Delta t_{2},\ldots,\Delta t_{k}\right)$
time steps, all pheromones in vertices and pair-trails included in
the $k$ regions cyclically return to the exact same decaying pheromone
levels (i.e. exactly the same temporal differences), we conclude that
a system in a configuratio\textit{n} \textit{$\mathcal{C}$,} will
transition at most $N$ distinct states to a configuration \textit{$\mathcal{C}_{N}$}
(on which the double visit condition is globally met) and then will
cyclically transition through $M$ states reaching, at each cycle,
a configuration \textit{$\mathcal{C}_{M}\sim\mathcal{C}_{N}$.}\qed
\end{proof}

Our next analysis is of the effect of multiple isolated branches on
future evolutions of a system with stable regions. Consider two scenarios,
both starting with the same configuration that has one isolated branch.
An arbitrary time later, a conquest creates another isolated branch,
the second branch being the same in both, only the \textit{time} of
its creation is different. Hence, there can be an arbitrarily large
time difference between the creation of the second isolated branch
in the two scenarios. Contemplating the case of an arbitrary number
of isolated branches created each at an arbitrary time, the complexity
of such presented scenarios may substantially grow. Nevertheless,
in term of system states the above complexity does not matter. Once
the decay of pheromones on an isolated branch is such that the double
visit condition is globally met, the conquest or rejoin threshold
is triggered, and afterwards no amount of further decay affects the
future evolutions of the system. This insensitivity holds regardless
of the presence of other isolated branches, simply because the double
visit is a \uline{local} condition, limited to the time difference
encoded in pheromones a on a vertex in a region and an adjacent vertex
on the branch. Thus, any two configurations that differ only by level
of pheromones on isolated branches for which the double visit condition
is globally met, are equivalent. Particularly, there must exist a
configuration such that the level of pheromones on isolated branches
is at its ``highest level'', i.e. the time marking on each vertex
of each isolated branch is the highest that allows the double visit
condition to be globally met. A branch with such ``highest level''
will have one vertex $v$, where $v=argmax\left(\varphi_{0}\left(w\right),\ w\in\textmd{isolated branch}\right)$
with a time-mark value $\varphi_{0}\left(v\right)=t-\left(M+N\right)$
where $t$ is the current time (i.e. a pheromone was left there $M+N$
steps before the current time), and all other vertices and edges with
(lower) values that agree with the ordered directions of pair-trails. 

Hence we conclude, again, that any configuration $\mathcal{C}$ that
includes multiple isolated branches will transition at most $M+N$
distinct states as long as all the regions remain stable.

\begin{theorem}
For a system with a graph $\mathcal{G}$ of stationary topology, and
$k$ agents, the set $\mathcal{S}$ of states is finite.\end{theorem}
\begin{proof}
Based on the above analysis, we conclude:
\begin{enumerate}
\item any configuration $\mathcal{C}$ is equivalent to a configuration
$\mathcal{C}_{P}$ identical to $\mathcal{C}$ except for having no
pheromone marking on edges that are not pair trails.
\item any configuration $\mathcal{C}_{P}$ that includes isolated branches
is equivalent to a configuration $\mathcal{C}_{A}$ identical to $\mathcal{C}_{P}$
except by the levels of pheromones on vertices and pair-trail edges
in those isolated branches that globally meet the double visit condition
on both configurations. Specifically, in $\mathcal{C}_{A}$, vertices
and pair-trail edges on each such isolated branch, will be of a ``highest
level'', i.e. will have one vertex $v$, where $v=argmax\left(\varphi_{0}\left(w\right),\ w\in\textmd{isolated branch}\right)$
with a time-mark value $\varphi_{0}\left(v\right)=t-\left(M+N\right)$
where $t$ is the current time (i.e. a pheromone was left there $M+N$
steps before the current time), $N=2*max\left(\Delta t_{1},\Delta t_{2},\ldots,\Delta t_{k}\right)$
and $M=lcm\left(\Delta t_{1},\Delta t_{2},\ldots,\Delta t_{k}\right)$
, and all other vertices and edges with values that agree to the directions
of pair-trails. 
\end{enumerate}
Therefore any configuration $\mathcal{C}$ is grouped in an equivalence
class with a correspondingly ``representative'' configuration $\mathcal{C}_{A}$.
To find out how many such classes exist, we observe that a $\mathcal{C}_{A}$
includes the following elements: $k$ regions with $k$ patrolling
routes, isolated branches that do not globally meet the double visit
condition and isolated branches of a highest level of time-markings
that globally meet the double visit condition. 

However, we have that:

\renewcommand{\labelenumi}{(\alph{enumi})} 
\begin{enumerate}
\item The number of possible choices of $k$ regions is finite (in a finite
graph).
\item For any arbitrary set of $k$ regions, the number of possible routes
in the $k$ regions is finite. 
\item For any arbitrary set of $k$ regions (with a particular selection
of $k$ routes) the number of vertices \uline{not} included in
these is finite, thus the number of possible isolated branches is
finite (and their possible assignments to whether they meet the double
visit condition or not is also finite).
\end{enumerate}
Therefore the set of possible representative configurations $\mathcal{C}_{A}$
is finite, each defines an equivalence class corresponding to a distinct
state of the Markov chain, hence the set $\mathcal{S}$ of system
states is finite. \qed
\end{proof}

Concluding the above analysis we see that in spite of the infinite
number of configurations possible for a system, the number of system
states, though quite large, is finite. Let us denote the finite set
of states of a system by \textit{S.}
\begin{definition}
(Gallager \cite{Stochastic processes book}), A Markov chain is an
integer-time process, $\left\{ X_{n},n>0\right\} $ for which the
sample values for each random variable $X_{n},n>1$ lie in a countable
set S and depend on the past only through the most recent random variable
$X_{n-1}$.
\end{definition}
Clearly, any state of the system at time $n$, formally represented
by $X_{n}\in S$, is dependent only on the previous state $X_{n-1}\in S$,
since our agents have no memory, and their decisions are based solely
on readings from vertices and edges of the configuration, which are
completely described by $X_{n-1}$. We can, therefore, analyze the
evolution of the system based on the theory of Markov chains. Our
aim is to prove that the Markov chain is not irreducible (i.e. given
enough time, the probability to reach some of its states tends to
zero), and that all its recurrent states represent balanced partitions.
To proceed with our analysis, we notice that the size of set $S$
grows very fast with the size of the graph. Calculations show that
even the simple example of Figure \ref{fig:4 first states} develops
to a surprisingly large chain. In order to be able to describe the
evolution of the system in a simple manner, we also define a \textit{partition}
of the environment.
\begin{definition}
The coloring of each vertex of a configuration $\mathcal{C}$ by its
patrolling agent along with the set of unvisited vertices form a partition
$\mathcal{P}$ of the graph. Partitions are unconcerned about the
levels of pheromones on the vertices and indifferent to agent locations,
thus only exhibit the regions of $\mathcal{C}$. 
\end{definition}

Many different configurations (and hence states too) correspond to
the same partition, therefore we can use the concept of a \textit{partition}
as an abstraction referring to all those configurations.

\begin{figure}[h]
\makebox[1\columnwidth]{%
\includegraphics[scale=0.7]{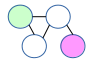}%
}\caption{\label{fig:float initial-partition}An example partition $\mathcal{P}$
of the graph}
\end{figure}

Figure \ref{fig:float initial-partition} is an example of a partition
of the environment graph that the system we discussed above arrived
to. From our previous discussion we know that it represents a set
of states of the underlying Markov chain. One characteristic of that
set of states is that it contains a cyclic path. This reflects the
fact that agents may cyclically repeat their patrolling route for
some period of time during which conquests or losses do not occur,
and the partition remains stationary. In fact, having a cyclic path
in the underlying Markov chain is characteristic of any reachable
partition.

Eventually conquests or losses are stochastically enabled leading
to a different partition, and, as a result, to a different set of
underlying states. In Figure \ref{fig:transition of partitions },
the system may remain in partition A for a while, as the underlying
chain cycles through the relevant states, but eventually it will probabilistically
transition to one of the partitions B,C,D,E. Note that the transition
from A to D means that both agents conquered one vertex each during
the same time-slot.

\begin{figure}[h]
\makebox[1\columnwidth]{%
\includegraphics[scale=0.5]{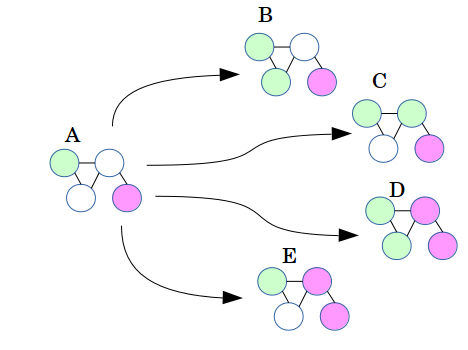}%
}\caption{\label{fig:transition of partitions }Transitions between partitions
(note that these transitions may correspond to many different transitions
between possible configurations or states, see figure \ref{fig:4 first states})}
\end{figure}

A \textit{recurrent class }in a Markov chain is a set of states which
are all accessible from each other (possibly passing through other
states), and no state outside the set is accessible (Gallager \cite{Stochastic processes book}).
The following Lemma shows that the set of states corresponding to
any balanced partition includes recurrent classes:
\begin{lemma}
\label{lem:C4}If a system remains in a balanced partition for a period
of time equal to twice the cover time of its largest region, it will
remain so indefinitely. \end{lemma}
\begin{proof}
We know from Lemma \ref{lem:A-region-is-consistent}, that if a region
remains unchanged for a period of time which is twice its cover time,
it becomes consistent, so the pheromone levels in all of its vertices
correctly indicate its cover time, $\varphi_{0}-\varphi_{1}=\Delta t$.
Therefore if the system remains in a balanced partition for a period
twice the largest cover time (the cover time of its largest region),
it is guaranteed that all the regions are consistent. Hence, we conclude
that no conquest attempts are subsequently possible, since the system
is balanced and conquest conditions can not be satisfied across any
border edge.\qed
\end{proof}

The conclusion of Lemma \ref{lem:C4} is that a balanced partition
with all its regions consistent, must correspond to a\textit{ recurrent}
(and \textit{periodic})\textit{ class }in the Markov chain\textit{.
}The random process continuously repeats a series of states based
on the individual agents' patrolling cycles. Since agents may reach
different patrolling routes for the same region, a balanced partition
may correspond to multiple recurring classes. Additionally we can
conclude the following: 
\begin{lemma}
\label{cor: C3-1}A system may enter a state in which the partition
is balanced, and then move into a state in which the partition is
not balanced. \end{lemma}
\begin{proof}
Clearly, we see while the conditions for consistency are not satisfied
for one or more regions in the partition, a conquest or loss may possibly
happen, hence the partition may become unbalanced.\qed
\end{proof}
Since recurrent classes exist, the stochastic transition matrix of
the Markov chain of the patrolling system must have the form:

\begin{minipage}[t]{1\columnwidth}%
\begin{center}
\vspace{2mm}
$\mathbf{M}=\left[\begin{array}{cc}
\mathbf{T} & \mathbf{TR}\\
\mathbf{0} & \mathbf{R}
\end{array}\right]$ 
\par\end{center}

\begin{center}
\medskip{}

\par\end{center}%
\end{minipage}\\
where $\mathbf{T}$ is a matrix of transitions between transient states
(for example those corresponding to non-balanced partitions), $\mathbf{R}$
is a matrix of transitions between states in recurrent classes (for
example those corresponding to balanced partitions with consistent
regions), and \textbf{TR }describes the transitions from transient
to recurrent states (for example those corresponding to balanced partitions
with inconsistent regions as described in corollary \ref{cor: C3-1}).
Note that we assume an initial distribution given by a row vector
$\bar{\pi}_{0}$ and hence the future distributions evolve according
to $\bar{\pi}_{k+1}=\bar{\pi}_{k}\mathbf{M}$. 

Each recurrent class representing one particular occurrence of $k$
routes in the $k$ regions contributes a section\textbf{ }$\mathbf{C_{i}}$
to \textbf{R}, of the form of a shifted identity matrix,\\
\begin{minipage}[t]{1\columnwidth}%
\begin{center}
\vspace{2mm}
$\mathbf{C_{i}}=\left[\begin{array}{ccccc}
0 & \mathbf{1} & 0 & \cdots & 0\\
\vdots &  & \mathbf{1} & 0 & \vdots\\
 & \cdots & 0 & \mathbf{1} & 0\\
0 &  & \cdots & 0 & \mathbf{1}\\
\mathbf{1} & 0 &  & \cdots & 0
\end{array}\right]$ 
\par\end{center}

\begin{center}
\medskip{}

\par\end{center}%
\end{minipage}\\
The rank of $\mathbf{C_{i}}$ is $M=lcm\left(\Delta t_{1},\Delta t_{2},\ldots,\Delta t_{k}\right)$
(i.e. the least common multiple of the cover times of the \textit{k
}regions), and is a function of the sizes of the regions in the partition.
The contribution of a particular balanced partition would be a matrix
$\mathbf{P_{j}}$ that comprises of a set of $\mathbf{C_{i}}$\textbf{
}matrices on its diagonal,\\
\begin{minipage}[t]{1\columnwidth}%
\noindent \begin{center}
\textbf{
\begin{eqnarray*}
\mathbf{P_{j}} & = & \left[\begin{array}{ccccc}
\mathbf{C_{1}} & 0 & \cdots &  & 0\\
0 & \mathbf{C_{2}}\\
\vdots &  & \mathbf{C_{3}} &  & \vdots\\
 &  &  & \ddots & 0\\
0 & \cdots &  & 0 & \mathbf{C_{r}}
\end{array}\right]
\end{eqnarray*}
}
\par\end{center}%
\end{minipage}\\
 where $r$ here is the finite number of possible route combinations
in the regions that form the partition. Therefore our goal is to show
that the structure of \textbf{R} is\\
\begin{minipage}[t]{1\columnwidth}%
\begin{center}
\vspace{2mm}
$\mathbf{R}=\left[\begin{array}{ccccc}
\mathbf{P_{1}} & 0 & \cdots &  & 0\\
0 & \mathbf{\mathbf{P_{2}}}\\
\vdots &  & \mathbf{P}_{\mathbf{3}} &  & \vdots\\
 &  &  & \mathbf{\ddots} & 0\\
0 & \cdots &  & 0 & \mathbf{P_{n}}
\end{array}\right]$\medskip{}

\par\end{center}%
\end{minipage}\\
 \\
 listing contributions from a finite number $n$ of possible balanced
partitions. To achieve this we must prove that recurring classes that
are \uline{not} representing balanced partitions do not exist.
This is done next.

To visualize the problem, we will classify all possible partitions
as shown in Table 1, recalling that a partition $\mathcal{P}$ represents
a set of states in the underlying Markov chain:

\begin{table}[h]
\caption{Classification of Partitions }

\begin{longtable}{cc>{\raggedright}p{4cm}>{\raggedright}p{6cm}}
\textbf{Partition Class} &  & \textbf{Name} & \textbf{Description}\tabularnewline
$\mathcal{CL}1$ &  & Not Covered & Any partition $\mathcal{P}$ that includes at least one vertex that
is \uline{not} part of a patrolling cycle.\tabularnewline
$\mathcal{CL}2$ &  & Covered, 

Not Balanced & Any partition $\mathcal{P}\notin\mbox{\ensuremath{\mathcal{CL}}1}$
that includes two adjacent regions with size difference greater than
one.\tabularnewline
$\mathcal{CL}3$ &  & Balanced, Unstable & Any partition $\mathcal{P}\notin\mbox{\ensuremath{\left\{  \mathcal{CL}1,\mathcal{CL}2\right\} } }$
where one or more inconsistent regions.\tabularnewline
$\mathcal{CL}4$ &  &  Balanced, Stable (Convergence) & Any partition $\mathcal{P}\notin\mbox{\ensuremath{\left\{  \mathcal{CL}1,\mathcal{CL}2\right\} } }$
where all regions are consistent. \tabularnewline
 &  &  & \tabularnewline
\end{longtable}
\end{table}

By construction of the classifications $\mathcal{CL}1,\mathcal{CL}2,\mathcal{CL}3,\mathcal{CL}4$
it is clear that these are mutually exclusive as well as complete,
hence they divide the set of all possible partitions, so that any
partition $\mathcal{P}$ can belongs to one and only one of the four
defined classes. It also means that any state of the system $X\in S$
can be classified to one and only one of the above classes. $\mathcal{CL}4$
includes the set of partitions that comply with Lemma \ref{lem:C4},
and $\mathcal{CL}3$ includes the partitions of Lemma \ref{cor: C3-1}.
Our goal is to show that from any initial state $X\in S$ the system
will reach some state in $\mathcal{CL}4$.

Figure \ref{fig:system modes}, which we will soon justify, depicts
the possible \textit{transitions} between the classes, i.e. from a
state in one class, exists a transition sequence in the Markov chain
to a state in another class as depicted in the diagram. 

\bigskip{}

\begin{figure}[h]
\makebox[1\columnwidth]{%
\includegraphics[scale=0.4]{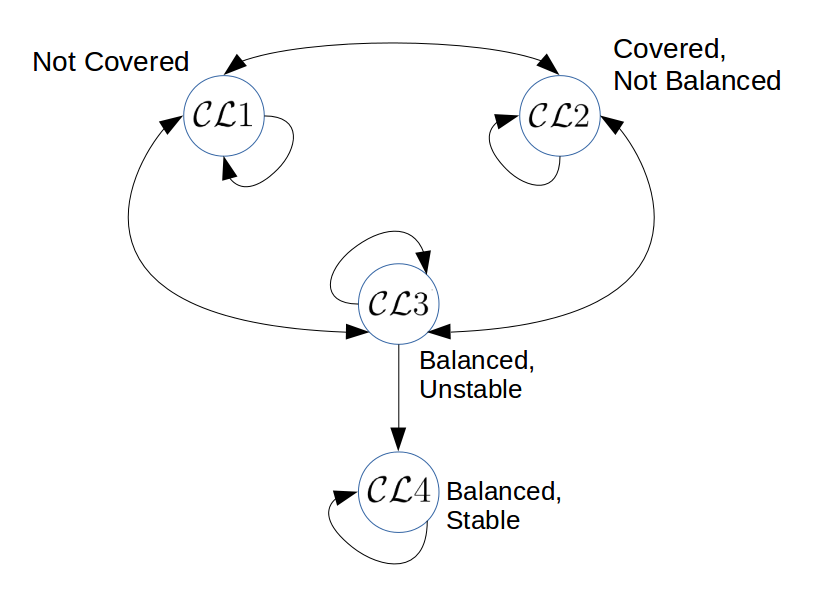}%
}\caption{\label{fig:system modes}Partition Class Transitions Diagram}
\end{figure}

\medskip{}

We first observe that clearly, any state in $\mathcal{CL}1$ or $\mathcal{CL}2$
must be \textit{transient}, the system cannot remain in any one of
them indefinitely. In $\mathcal{CL}1$ the partition includes free
vertices (that are not part of any region). For an agent that persist
in visiting a neighboring vertex, conquest conditions will be eventually
met, and the agent will make repeated attempts to conquer it. Eventually
all free vertices will be conquered and the graph becomes covered.
Similarly, on $\mathcal{CL}2$, when an agent visits a vertex bordering
a larger region there is a chance it will conquer the neighboring
vertex. The conquest could cause the pruning of the spanning tree
and even a balloon explosion of the larger region, and the system
shifts to $\mathcal{CL}1$, or maybe, by chance, it will move to a
balanced partition and end up in $\mathcal{CL}3$.

In $\mathcal{CL}3$ the system is balanced but inconsistent. We know
from Lemma \ref{lem:C4} that it can become consistent, thus shift
to $\mathcal{CL}4$, since there is always a chance that no conquests
or losses will occur in any finite period of time. If a conquest does
occur, it may become unbalanced, and then we are back to $\mathcal{CL}1$
or $\mathcal{CL}2$. So any state in $\mathcal{CL}3$ is transient
and may sink to a state in $\mathcal{CL}4$.

We are left to show that from any state in $\mathcal{CL}1$ and $\mathcal{CL}2$
there exists a path to $\mathcal{CL}3$. A sequence of states that
switches the system back and forth between $\mathcal{CL}1$ and $\mathcal{CL}2$
is the scenario were the system repeatedly evolves to a covered graph
only to retract by events such as a balloon explosion. We will show
now that from any configuration there is a strictly positive probability
to find a balanced partition, hence move to $\mathcal{CL}3$, and
as result there is a strictly positive probability to get into $\mathcal{CL}4$.
This means that $\mathcal{CL}4$, the set of balanced partitions,
are the only partitions that map to recurrent classes. To do that
we recall our restriction to Hamiltonian graphs (or \textit{k-}connected
graphs) and invoke the property of vertex loss.
\begin{lemma}
\label{lem: C1-C2 are transient}Given a system with an environment
graph $\mathcal{G}$ having a Hamiltonian path (or a k-connected graph),
and $n$ agents, in an arbitrary configuration $\mathcal{C}_{0}$,
there is a strictly positive probability for the system to evolve
to a balanced partition. \end{lemma}
\begin{proof}
Let us assume the contrary, that there exists a configuration $\mathcal{C}_{0}$
representing a state $s_{0}\in\left\{ \mathcal{CL}1,\mathcal{CL}2\right\} $
such that the probability of any arbitrarily long step sequence starting
at $\mathcal{C}_{0}$ to arrive into a state in $\mathcal{CL}3$ is
strictly zero. We first note that there is a strictly positive probability
that the next change of a region is the loss of a vertex. Namely,
assume that for a sequence of configurations $\left\{ \mathcal{C}_{0},\mathcal{C}_{1},\ldots,\mathcal{C}_{N}\right\} $,
where $N$ is finite, at configurations $\left\{ \mathcal{C}_{0},\mathcal{C}_{1},\ldots,\mathcal{C}_{N-1}\right\} $
there were no changes to any of the regions in the partition of the
environment, and at $\mathcal{C}_{N}$ one or more of the regions
lost one vertex each, following failed conquest attempts. This could
repeat and consequently, there is a strictly positive (though very
small) probability for a sequence of configurations $\left\{ \mathcal{C}_{0},\mathcal{C}_{1},\ldots,\mathcal{C}_{M}\right\} $,
where $M$ is finite, to arrive to a configuration $\mathcal{C}_{M}$
where each of the regions is of size 1. Now, we note that there is
a strictly positive probability that the next change of a region is
a conquest launched from one vertex to a neighbor along a given Hamiltonian
path. This could repeat until the regions form one of the possible
balanced partitions along the Hamiltonian path and the system is now
in $\mathcal{CL}3$. This evolution contradicts our assumption regarding
$\mathcal{C}_{0}$, and we therefore conclude that $\left\{ \mathcal{CL}1,\mathcal{CL}2\right\} $
is a transient set, thus from any arbitrary configuration there is
a strictly positive probability to arrive to a balanced partition.

Note that a similar argument can be made for a \textit{k}-connected
environment graph.\qed

We therefore conclude: \end{proof}
\begin{theorem}
A system with a Hamiltonian or k-connected graph $\mathcal{G}$ with
stationary topology, and n agents implementing AntPaP converges in
finite expected time to a balanced and stable partition with probability
1.\end{theorem}
\begin{proof}
The conclusion from Lemma \ref{lem: C1-C2 are transient}, is that
any state $s\in\left\{ \mathcal{CL}1,\mathcal{CL}2\right\} $ is transient.
We also know that any state $s\in\mathcal{CL}3$ is transient by Lemma
\ref{lem:C4}. This means that all recurrent classes included in the
matrix $\mathbf{R}$ of the stochastic matrix $\mathbf{M}=\left[\begin{array}{cc}
\mathbf{T} & \mathbf{TR}\\
\mathbf{0} & \mathbf{R}
\end{array}\right]$ represent balanced and stable partitions, classified as $\mathcal{CL}4$.
A Markov chain described by a stochastic matrix of this form, will
eventually enter a recurrent state, regardless of the initial state,
and the probability that this takes more than $t$ steps approaches
zero geometrically with $t$ (see, for example, Gallager \cite{Stochastic processes book}).
We conclude that a system with a graph \textit{$\mathcal{G}$ }(of\textit{
}stationary\textit{ }topology) and $n$ agents implementing AntPaP
converges with probability 1 and a finite expected time to a balanced
and stable partition.\qed
\end{proof}

\pagebreak{}

\section{Experimental Results and Discussion}

We presented and thoroughly analyzed the AntPaP algorithm for continuously
patrolling a graph-environment with simple finite state automaton
agents (or bots) using ``pheromone traces''. The simulations presented
so far were on an environment in the shape of a square. On such an
environment, we know that many balanced partitions do indeed exist.
Practical scenarios are seldom so simplistic. In many important cases,
the environment graph is, in fact, uncharted and much more complex
in its structure. Still, agents implementing AntPaP will find a balanced
partition with probability one (almost surely), if such a partition
exists, and will certainly divide their work fairly even when such
partitions do not exist.

The shape of the environment considerably affects the time to convergence.
The number of balanced partitions that the environment graph has is,
naturally, one of the major factors. So is their diversity, i.e. how
different the balanced partitions are from each other. If the balanced
partitions are similar to one another, the dependency of the time
to convergence on the initial locations of the agents tends to be
higher than if the solutions are further apart. 

Consider the system of Figure \ref{fig:T with odd}. The initial positions
of 7 agents are shown in the first snapshot at $t=1$. Next, at $t=896$,
the lower section becomes almost covered. At $t=5995,$ the upper
section is almost covered, and the two agents there clearly have larger
regions than the agents in the lower section. At $t=12993$, the cyan
agent is trapped in the upper section, and we witness a competition
between the agents from the lower section to grow their regions into
the prolonged section, that the cyan agent abandoned.

\begin{figure}[H]
\includegraphics[scale=0.25]{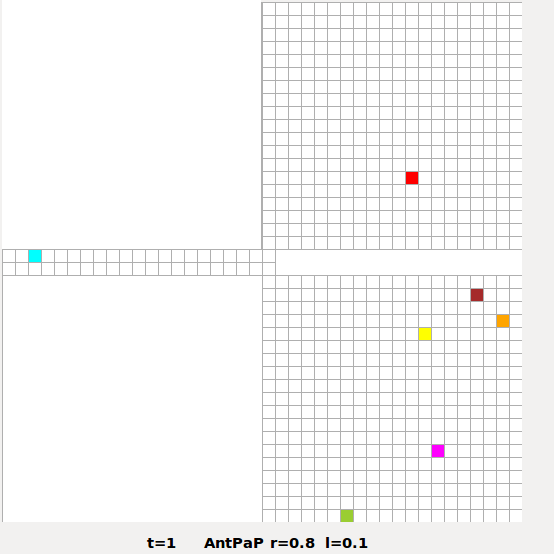}\includegraphics[scale=0.25]{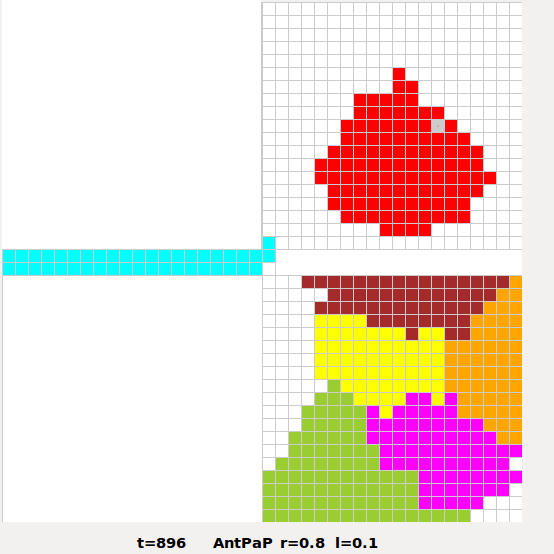}\includegraphics[scale=0.25]{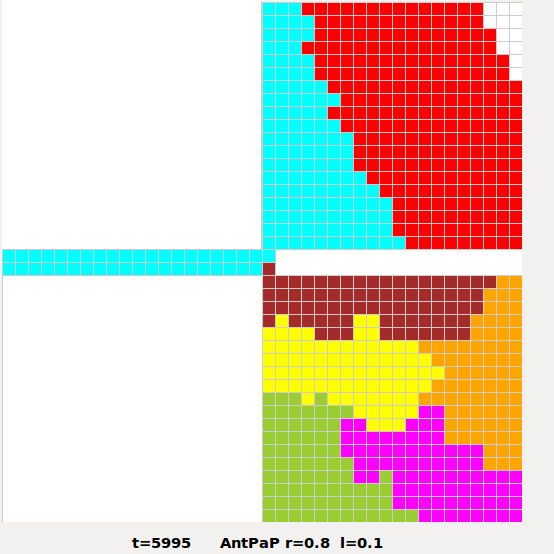}

\bigskip{}

\begin{centering}
\includegraphics[scale=0.25]{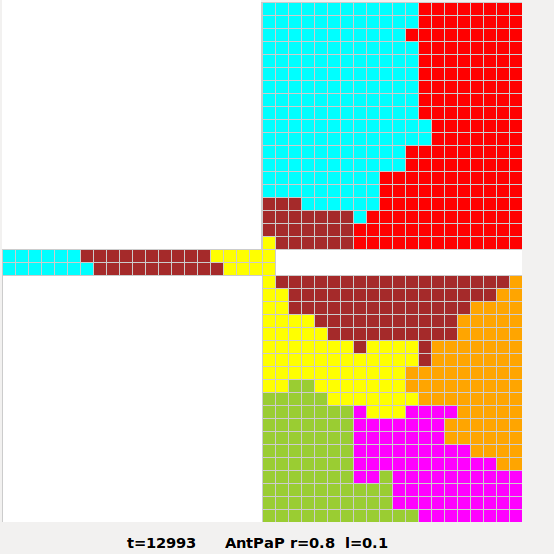}\includegraphics[scale=0.25]{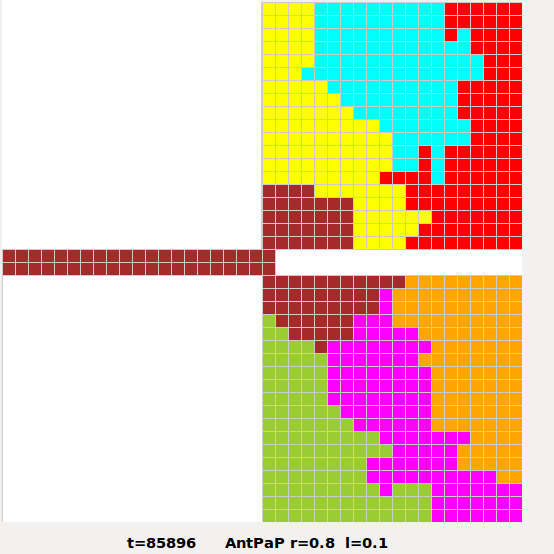}\includegraphics[scale=0.25]{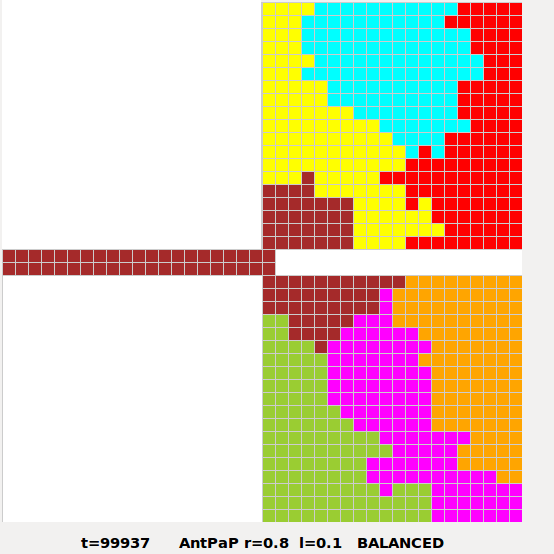}
\par\end{centering}

\caption{\label{fig:T with odd}Evolution of a system with 7 a(ge)nts}
\end{figure}

At $t=85896$, the competition ends after the yellow agent traversed
into the upper section. Now we have 3 agents on each of the upper
and lower sections, and one on the prolonged section. Soon after,
the system reaches a balanced partition. Clearly, there are many balanced
partitions for this system, but all of them have one agent on the
prolonged section, and 3 agents on each of the upper and lower sections.
Initial conditions with 3 agents on the upper and lower sections each
will ensure faster convergence to a balanced partition. Following
this experiment and discussion, it is interesting to consider a system
with the same environment graph and an \uline{even} number of agents.

We simulated a system with the same environment graph and 2 agents.
This system has only one balanced partition, shown in Figure \ref{fig:The-only-balanced}. 

\begin{figure}[H]
\makebox[1\columnwidth]{%
\includegraphics[scale=0.25]{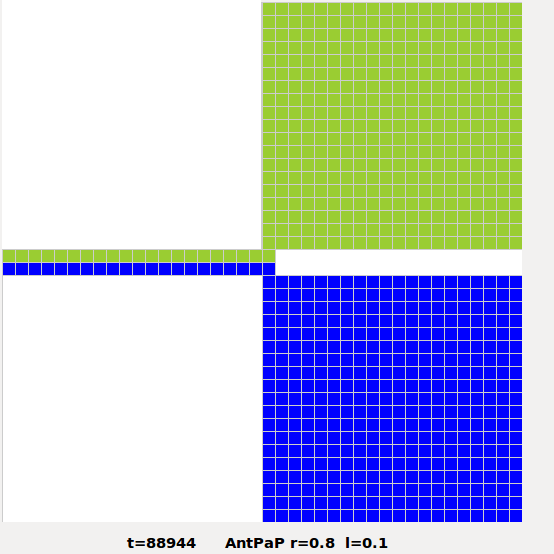}%
}\caption{\label{fig:The-only-balanced}The only balanced partition of a systen
with 2 agents}
\end{figure}

Since only one balanced partition exists, it is reasonable to predict
that the required time for convergence might be substantial. Figure
\ref{fig:T with 2} shows snapshots of an evolution of this system.
Both the violet and yellow agents are initially located in the lower
section. After a while, the violet agent expands its region so that
part of it extends into the upper section. A while later, the violet
region covers almost all of the upper section as well as the prolonged
section. Then, the yellow agent begins to expand into the prolonged
section, eventually causing a ``balloon explosion'' of violet's
region. Soon enough, the violet agent responds, and causes a balloon
explosion of the yellow's region. 

\begin{figure}[H]
\includegraphics[scale=0.25]{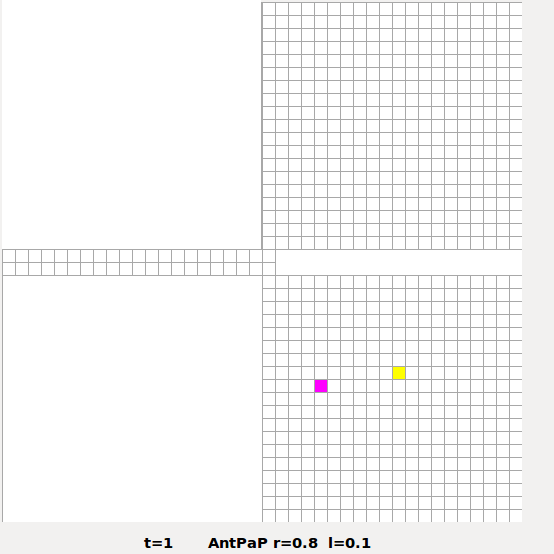}\includegraphics[scale=0.25]{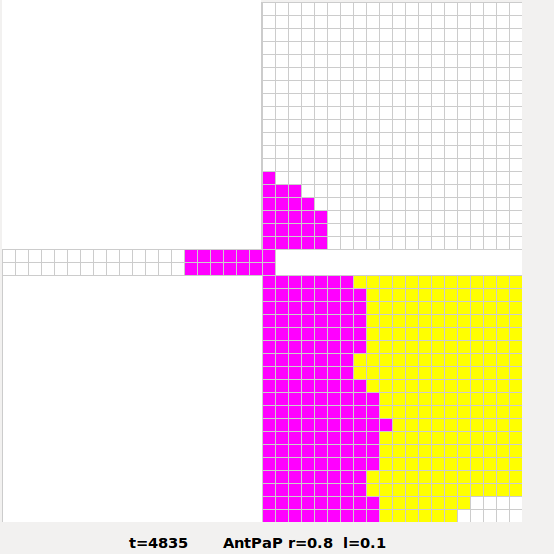}\includegraphics[scale=0.25]{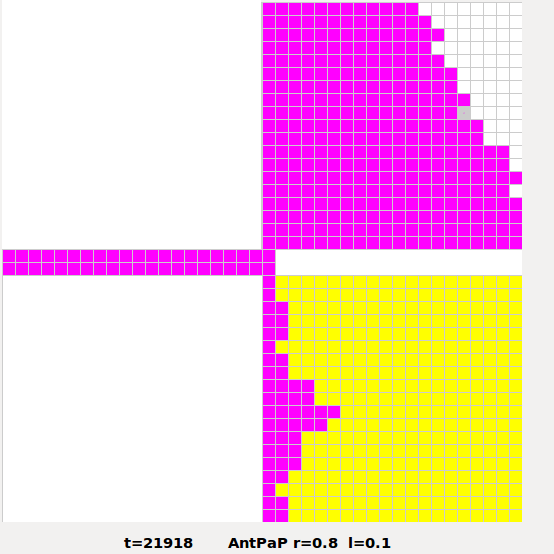}

\bigskip{}

\begin{centering}
\includegraphics[scale=0.25]{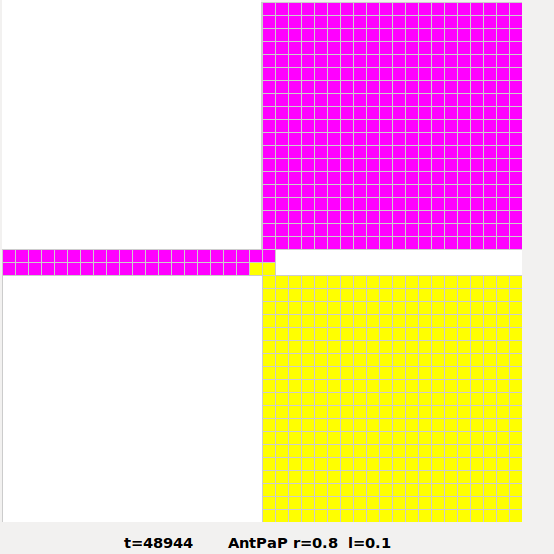}\includegraphics[scale=0.25]{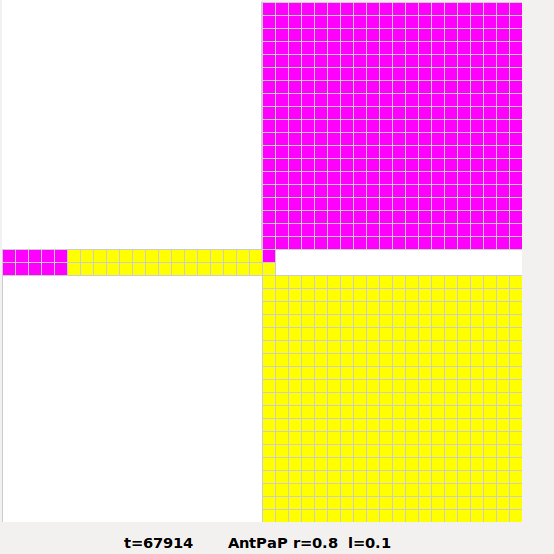}\includegraphics[scale=0.25]{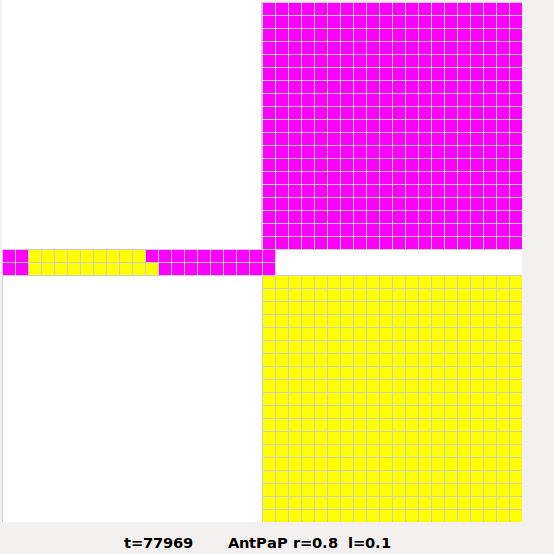}
\par\end{centering}

\caption{\label{fig:T with 2}Evolution of a system with 2 a(ge)nts}
\end{figure}

Due to the shape of the graph, this cycle may repeat over and over
again. It will stop only when the single possible balanced partition
is reached, and subsequently the regions ``lock-in'', and the system
remains stable. For that to happen, an agent must conquer the appropriate
half of the vertices of the prolonged section. We know that this will
eventually happen with probability 1, however the time it will take
can be very very long.

\begin{figure}[H]
\makebox[1\columnwidth]{%
\includegraphics[scale=0.25]{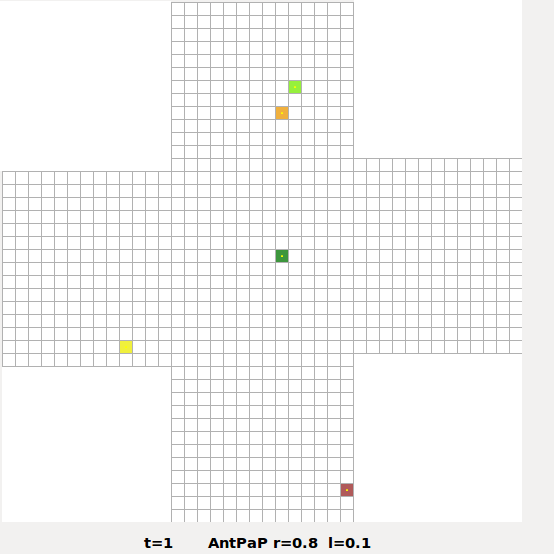}\includegraphics[scale=0.25]{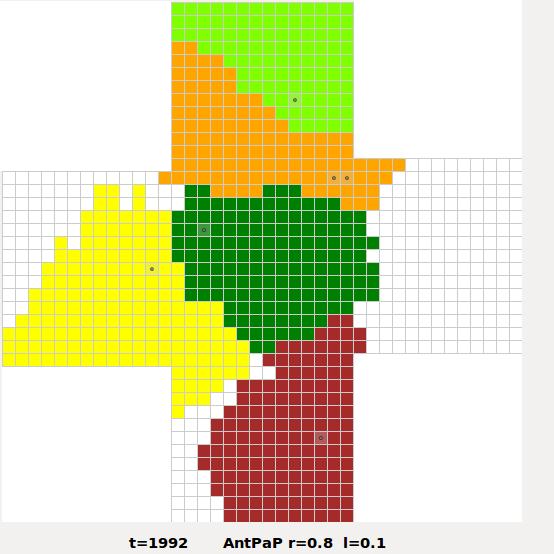}\includegraphics[scale=0.25]{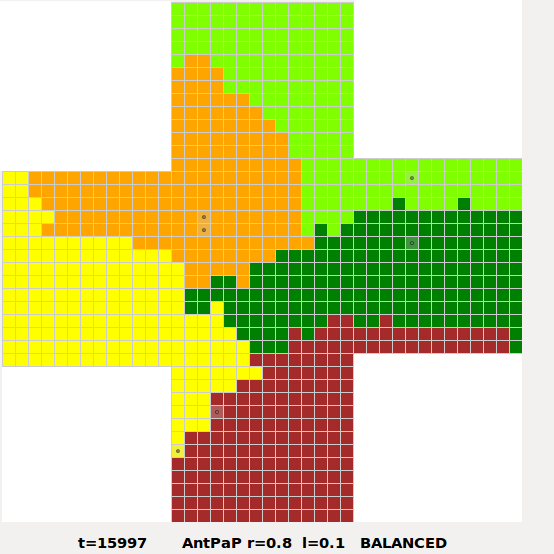}%
}\caption{\label{fig:PLUS w 5 agents}A system with a ``Cross'' graph and
5 agents}
\end{figure}

\begin{figure}[H]
\makebox[1\columnwidth]{%
\includegraphics[scale=0.25]{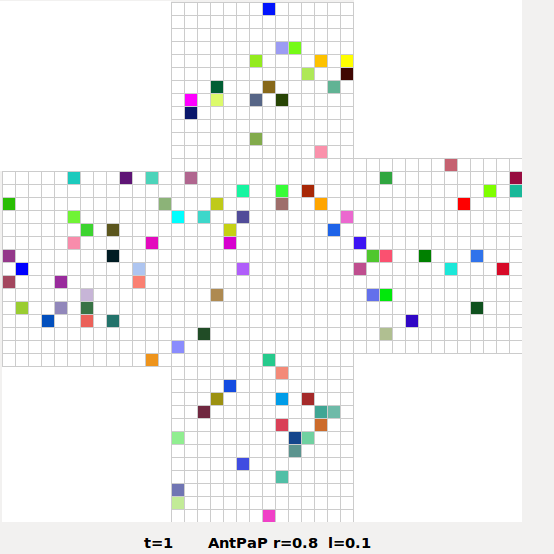}\includegraphics[scale=0.25]{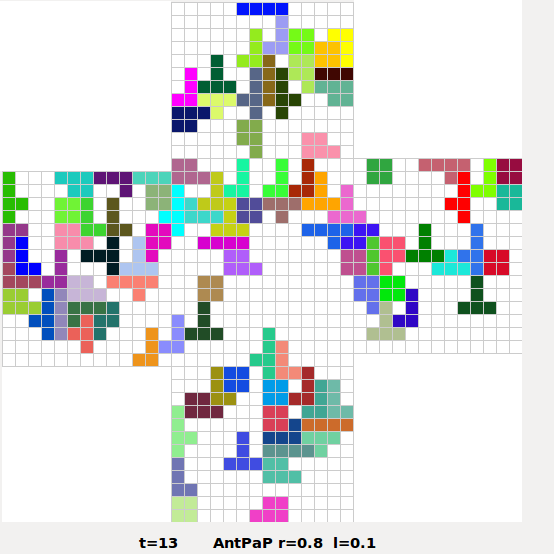}\includegraphics[scale=0.25]{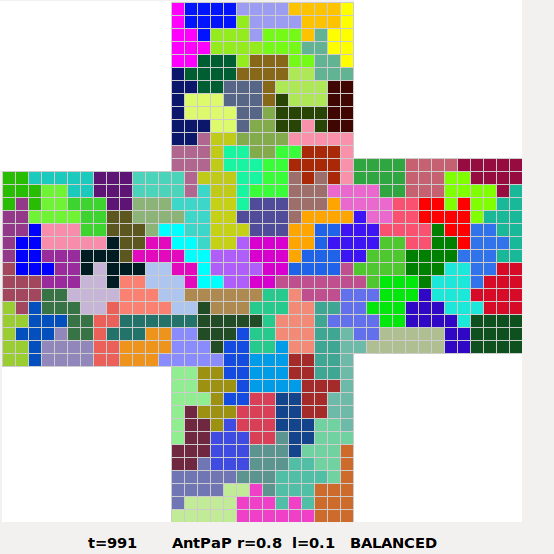}%
}\caption{\label{fig:PLUS w 100 agents}A system with the same ``Cross'' graph
as in Figure \ref{fig:PLUS w 5 agents}, and 100 agents}
\end{figure}

The number of agents is also an important factor of convergence time.
Generally, more agents hasten the convergence. Figure \ref{fig:PLUS w 5 agents}
shows an evolution of a system with 5 agents on a different environment.
We shall call this environment graph the ``Cross''. Figure \ref{fig:PLUS w 100 agents}
shows a system with the same ``Cross'' graph and 100 agents. Here,
the ``pressure'' that an agent ``feels'' from other ``balloons''
quickly accumulates around its region, and the convergence is swift.
Figure \ref{fig:PLUS agents vs. convergence} depicts results of multiple
simulation runs, of systems with the ``Cross'' graph of figure \ref{fig:PLUS w 5 agents},
exhibiting convergence time as a function of the number of agents.

\begin{figure}[H]
\makebox[1\columnwidth]{%
\includegraphics[scale=0.4]{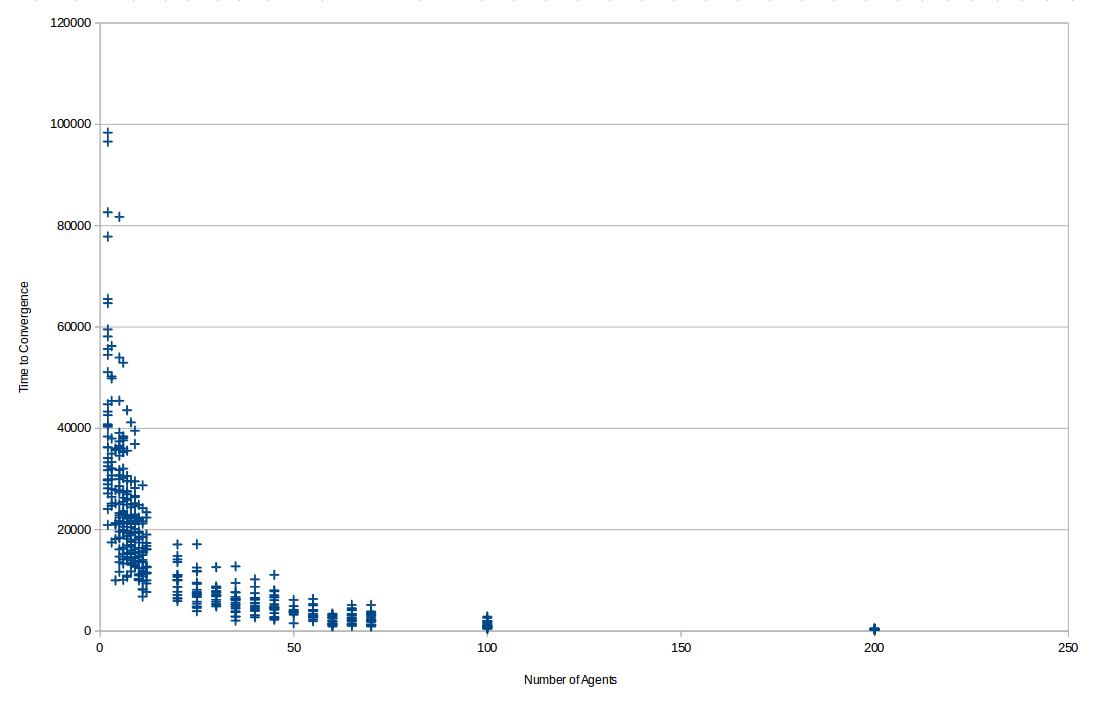}%
}\caption{\label{fig:PLUS agents vs. convergence}Convergence time of a system
with a ``Cross'' graph as function of the number of agents }
\end{figure}

In some systems, particular numbers of agents may cause a substantially
larger time to convergence. In Figure \ref{fig:ROOMS balanced} we
present a balanced partition in a graph environment that we call ``6
Rooms''. 

\begin{figure}[H]
\makebox[1\columnwidth]{%
\includegraphics[scale=0.3]{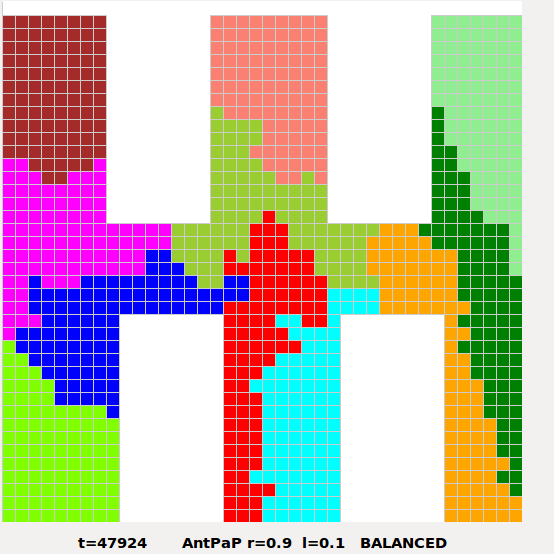}%
}\caption{\label{fig:ROOMS balanced}Balanced partition reached on a ``6 Rooms''
graph }
\end{figure}

Systems with a ``6 Rooms'' graph and 6 agents sometimes require
a substantially longer convergence time, as shown in Figure \ref{fig:ROOMS agents vs. convergence}.
Ignoring the outliers at 6 agents, Figure \ref{fig:ROOMS agents vs. convergence No outliers}
shows that the chart exhibiting convergence time as a function of
the number of agents is similar in shape to the one we have seen for
the ``Cross'' graph, in Figure \ref{fig:PLUS agents vs. convergence}.

\begin{figure}[H]
\makebox[1\columnwidth]{%
\includegraphics[scale=0.5]{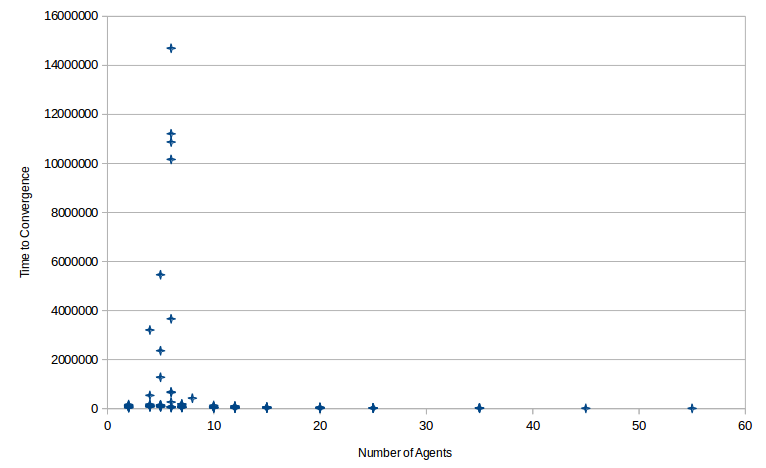}%
}\caption{\label{fig:ROOMS agents vs. convergence}Convergence time of a system
with a ``6 Rooms'' graph as function of the number of agents}
\end{figure}
\begin{figure}[H]
\makebox[1\columnwidth]{%
\includegraphics[scale=0.5]{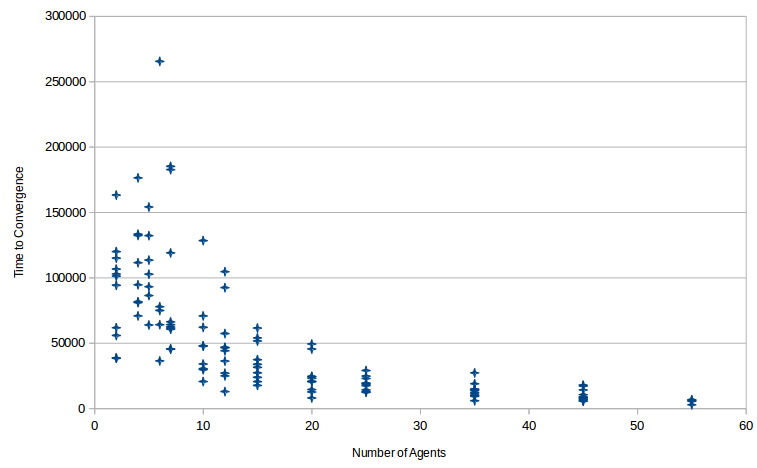}%
}\caption{\label{fig:ROOMS agents vs. convergence No outliers}Zoom-In on the
chart of Figure \ref{fig:ROOMS agents vs. convergence} }
\end{figure}

\begin{figure}[H]
\makebox[1\columnwidth]{%
\includegraphics[scale=0.4]{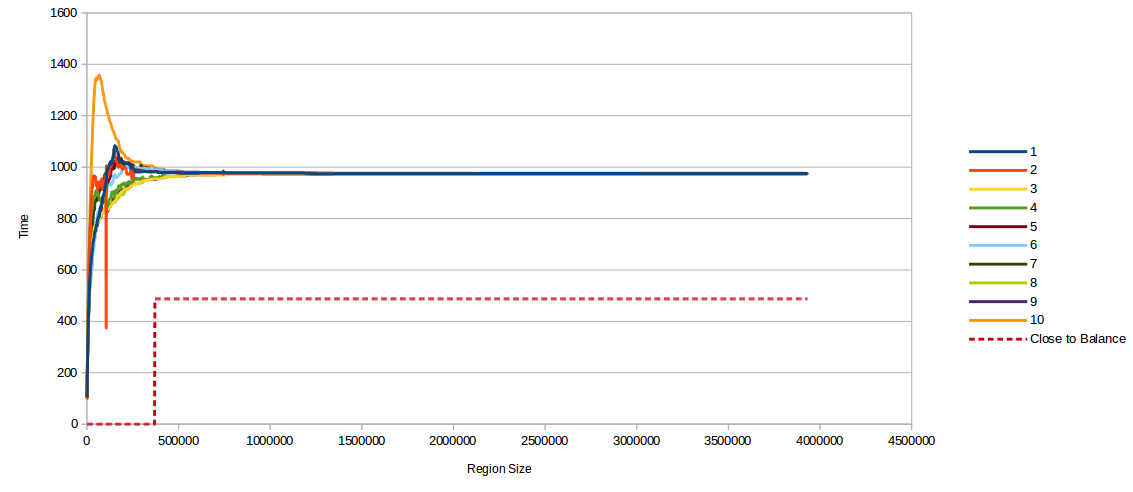}%
}\caption{\label{fig:ROOMS region size evoultion}An evolution of a system of
10 agents and a ``6 Rooms'' graph of about 10,000 vertices. Convergence
was achieved at approximately $t=4,000,000$. But the system became
``close to balanced'' rather quickly.\protect \\
}
\end{figure}

\begin{figure}[H]
\makebox[1\columnwidth]{%
\includegraphics[scale=0.3]{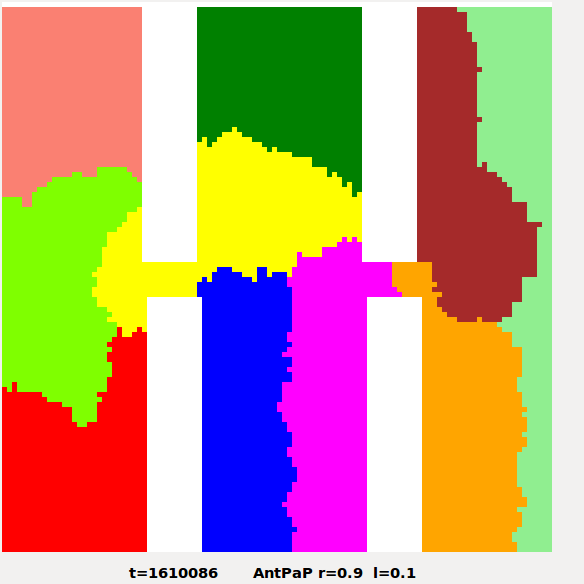}\includegraphics[scale=0.3]{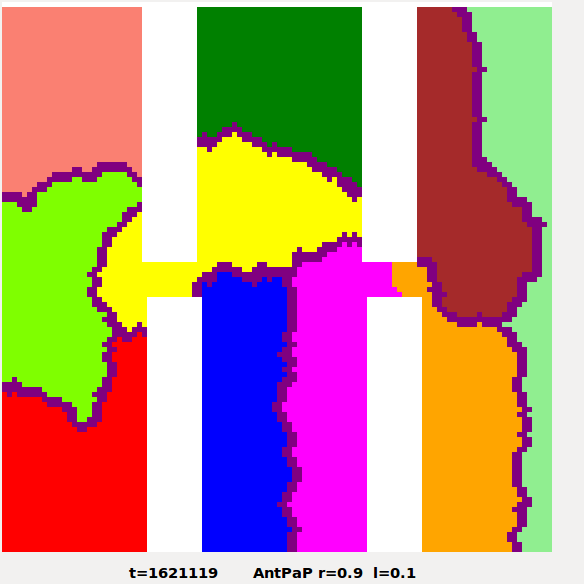}%
}\caption{\label{fig:ROOMS balanced with contours}An evolution that reached
``close to balanced'' quickly, and remain so for long. \protect \\
}
\end{figure}
In the simulations described above, we tested the evolution of the
multi-agent patrolling process until convergence to a stable and balanced
partition. However remarkably, the system evolves rather quickly to
close-to-balanced partitions due to the ``balloon'' forces implicitly
driving the agents' behavior. Therefore, for practical purposes we
see that the AntPaP algorithm balances the work of the agents much
earlier than its convergence time, and the partitioning becomes reasonably
good rather quickly. This property is crucial in case of time varying
topologies. Hence, AntPaP is a versatile and adaptive process. Considering
again the ``6 Rooms'' example with 10 agents, we see in Figure \ref{fig:ROOMS region size evoultion}
a temporal evolution of AntPaP until a stable and balanced partition
is achieved. As is clear on the chart displaying the time evolution
of the sizes of the 10 regions, the system reached convergence at
approximately $t=4,000,000$ steps. However it is also clear that
after approximately $400,000$ steps, the difference between the largest
and smallest regions in the partition of the environment graph is
already insignificant. In the chart, a system is defined as ``close
to balanced'' when more than 99\% of the graph is covered, and the
difference between the largest and smallest regions is less than 5\%
the ideally balanced size (i.e. the graph size divided by the number
of agents). Figure \ref{fig:ROOMS balanced with contours} exhibits
a partition reached when the system was ``close to balanced''. Both
snapshots show the same partition (at two different times). The snapshot
at the right also shows the borders between regions that ``reached
balance'' (i.e. their size difference is at most 1 vertex) depicted
in purple. There is only one border which is not yet balanced, between
the magenta region and the dark yellow region, located in the right
``corridor''. These regions are close in their sizes, and as a result
the double visit condition does not occur very often. Despite the
partition not being balanced yet, the division of work between agents
is already fair, hence for practical purposes, a ``close to balanced''
condition is good enough.

We note in summary that AntPaP does not produce \textit{k-cut }partitions
\cite{k-cut}, and generally assumes that there are no constraints
on the grouping of vertices. Some important real-world problems impose
such constraints, for example, the allocation of users in a social
network to hosting servers, according to their interconnections. Other
real-world problems, however, do not impose such constrains, for example,
the division of work patrolling the world-wide-web for content analysis
and classification. In view of the good properties discussed above,
we envision that AntPaP could become a building block for distributed
algorithms aiming to fairly divide between agents the labor of patrolling
an environment, using very simple agents constrained to local interactions
based on tiny ``pheromone'' marks left in the environment. 

\newpage{}

\pagebreak{}

\pagebreak{}


\begin{thebibliography}{10}
\bibitem{BDFS} Y. Elor, A. M. Bruckstein, \textquotedbl{}Multi-a(ge)nt
graph patrolling and partitioning\textquotedbl{}. Proceedings of the
2009 IEEE/WIC/ACM International Joint Conference on Web Intelligence
and Intelligent Agent Technology-Volume 02. IEEE Computer Society,
2009.

\bibitem{MLDFS}I.A. Wagner, M. Lindenbaum, A. M. Bruckstein, ``Efficiently
searching a graph by a smell-oriented vertex process''. Annals of
Mathematics and Artificial Intelligence, 24(1-4):211\textendash{}223,
1998.

\bibitem{Even}S. Even, ``Graph Algorithms''. Rockville, MD: Comput.
Sci. Press, 1979

\bibitem{Clustering Survey}P. Berkhin, \textquotedbl{}A survey of
clustering data mining techniques.\textquotedbl{} Grouping multidimensional
data. Springer Berlin Heidelberg, pp. 25-71, 2006.

\bibitem{VLSI}C. J. Alpert, A. B. Kahng, ``Recent directions in
netlist partitioning: a survey'', Integration, the VLSI Journal ,
19(1\textendash{}2):1-81,1995

\bibitem{Load Balancing}K. Schloegel, G. Karypis, V. Kumar, ``Parallel
static and dynamic multi-constraint graph partitioning'', Concurrency
and Computation: Practice and Experience, 14 (3):219\textendash{}240,2002 

\bibitem{k-cut}R. G. Downey, V. Estivill-Castro, M. R. Fellows, E.
Prieto, F. A. Rosamond, ``Cutting Up Is Hard To Do: the Parameterized
Complexity of k -Cut and Related Problems'', Electronic Notes in
Theoretical Computer Science, 78():209\textendash{}222,2003

\bibitem{DFS Tarjan Spanning Tree}R. Tarjan, ``Depth-First Search
and Linear Graph Algorithms'', SIAM, 1(2):146-160, 1972

\bibitem{bee colony}J. D. McCaffrey. ``Graph Partitioning using
a Simulated Bee Colony Algorithm'', IEEE Conference on Information
Reuse and Integration (IRI) Las Vegas, USA, pp. 400-405, 2011

\bibitem{aNTS}F. Comellas, E. Sapena, ``A multiagent algorithm for
graph partitioning''. EvoWorkshops, Lecture Notes in Computer Science
3907:279\textendash{}285, 2006

\bibitem{multi-agent patrolling strategies}Y. Chevaleyre, F. Sempe,
G. Ramalho, \textquotedbl{}A theoretical analysis of multi-agent patrolling
strategies.\textquotedbl{} Proceedings of the Third International
Joint Conference on Autonomous Agents and Multiagent Systems-Volume
3. IEEE Computer Society, pp. 1524-1525, 2004.

\bibitem{ACO applied to parolling}F. Lauri, F. Charpillet, \textquotedbl{}Ant
colony optimization applied to the multi-agent patrolling problem.\textquotedbl{}
IEEE Swarm Intelligence Symposium, 2006.

\bibitem{ACO Ant System}M. Dorigo, V. M. Maniezzo, A. Colorni, \textquotedbl{}Ant
system: optimization by a colony of cooperating agents.\textquotedbl{}
IEEE Transactions on Systems, Man, and Cybernetics, Part B (Cybernetics)
26(1): 29-41, 1996.

\bibitem{Gyori}E. Gyori, \textquotedbl{}On division of graphs to
connected subgraphs.\textquotedbl{} Combinatorics (Proc. Fifth Hungarian
Colloq., Keszthely). Vol. 1, 1976.

\bibitem{Stochastic processes book}R. G. Gallager, Stochastic processes:
theory for applications. Cambridge University Press, 2013.

\bibitem{Nerode}A. Nerode, \textquotedbl{}Linear automaton transformations.\textquotedbl{}
Proceedings of the American Mathematical Society 9(4): 541-544, 1958\end{thebibliography}
\end{document}